\documentclass[12pt, draftclsnofoot, onecolumn]{IEEEtran}
\usepackage{algorithm}
\usepackage{algpseudocode}
\usepackage{pifont}

\usepackage{cite}
\usepackage{amsmath}%
\usepackage{times}
\usepackage[utf8]{inputenc}
\usepackage{authblk}
\usepackage{textcomp}
\usepackage{amsfonts}%
\usepackage{amssymb}%
\usepackage{amsthm}
\usepackage{cancel} % to add a correction
\usepackage{graphicx}%
\usepackage{float}
\usepackage{layout}
\usepackage{array}%
\usepackage{subfig}%
\usepackage{comment}
\usepackage{array}%
\usepackage{color}%
\usepackage[usenames,dvipsnames]{xcolor}
\usepackage{soul}
\usepackage{footmisc}%
\theoremstyle{plain}
\usepackage{epstopdf}
\usepackage{amssymb}
\usepackage{bbm}
\usepackage{mathtools}
\newtheorem{theorem}{Theorem}

\newtheorem{claim}{Claim}

\newtheorem{proposition}{Proposition}

\newcommand{\g}{\gamma} % scaling factor of highly popular items
\newcommand{\E}{\mathbb{E}}

\newcommand{\cY}{\mathcal{Y}}

\newcommand{\al}{\alpha}

\newcommand{\I}{\mathbb{I}}

\newcommand{\tmu}{\tilde{\mu}}
\newcommand{\tR}{\tilde{R}}

\newcommand{\hy}{\hat{y}}
\newcommand{\bL}{\hat{L}}
\DeclareMathOperator*{\argmax}{arg\,max}
\DeclareMathOperator*{\argmin}{arg\,min}

\begin{document}

%\color{blue}

\title{Towards A Marketplace for Mobile Content: Dynamic Pricing and Proactive Caching}
%\author{\IEEEauthorblockN{Faisal Alotaibi,  Sameh Hosny, John Tadrous, Hesham El Gamal and Atilla Eryilmaz} \\
%        \IEEEauthorblockA{Department of Electrical and Computer Engineering\\
%        The Ohio State University, Columbus, Ohio, USA.\\
%        Email: \{alotaibi.12, hosny.1, tadrous.1, elgamal.2, eryilmaz.2\}@osu.edu}}

\author{
    \IEEEauthorblockN{F. Alotaibi\IEEEauthorrefmark{1}, S. Hosny\IEEEauthorrefmark{1}, J. Tadrous\IEEEauthorrefmark{2}, H. El Gamal\IEEEauthorrefmark{1} and A. Eryilmaz\IEEEauthorrefmark{1}}\\
    \IEEEauthorblockA{\IEEEauthorrefmark{1}The Ohio State University, USA.\\ Email: \{alotaibi.12, hosny.1, elgamal.2, eryilmaz.2\}@osu.edu}
    \\\IEEEauthorblockA{\IEEEauthorrefmark{2} Rice University, USA. Email: jtadrous@rice.edu}
}

% For example:
% ------------
%\address{School\\
%	Department\\
%	Address}
%
% Two addresses (uncomment and modify for two-address case).
% ----------------------------------------------------------
%\twoauthors
%  {A. Author-one, B. Author-two\sthanks{Thanks to XYZ agency for funding.}}
%	{School A-B\\
%	Department A-B\\
%	Address A-B}
%  {C. Author-three, D. Author-four\sthanks{The fourth author performed the work
%	while at \cdots}}
%	{School C-D\\
%	Department C-D\\
%	Address C-D}
%
%\ninept
%
\maketitle
%

% ===================================================================
% ============================= Abstract ============================
% ===================================================================

\begin{abstract}
In this work, we investigate the profit maximization problem for a wireless network carrier and the payment minimization for end-users. Motivated by recent findings on proactive resource allocation, we focus on the scenario whereby end-users who are equipped with device-to-device (D2D) communication can harness predictable demand in \emph{proactive data contents caching} and the possibility of trading their proactive downloads to minimize their expected payments. The carrier, on the other hand, utilizes a \emph{dynamic pricing} scheme to differentiate between off-peak and peak time prices and applies commissions on each trading process to further maximize its profit. A novel \emph{marketplace} that is based on risk sharing between end-users is proposed where the tension between carrier and end-users is formulated as a Stackelberg game. The existence and uniqueness of the non-cooperative sub-game Nash equilibrium is shown. Furthermore, we explore the equilibrium points for the case when the D2D is available and when it is not available, and study the impact of the uncertainty of users’ future demands on the system’s performance. In particular, we compare the new equilibria with the baseline scenario of flat pricing. Despite end-users connectivity with each other, the uncertainty of their future demands, and the freshness of the pre-cached contents, we characterize a new equilibria region which yields to a win-win situation with respect to the baseline equilibrium. We show that end-users’ activity patterns can be harnessed to maximize the carrier’s profit while minimizing the end-users expected payments.
\end{abstract}


% ===================================================================
% =========================== Introduction ==========================
% ===================================================================

%
\section{Introduction}
The tremendous increase in demand for spectrum-based services and devices has led network carriers these days to experience a major demand and supply mismatch during the whole day. Throughout the daytime, demand levels raise substantially reaching the network capacity and causing excessive service costs and undesirable congestions. On the other hand, at the nighttime, the demand significantly drops to minimum levels rendering the network resources severely underutilized \cite{1_oss.oetiker.ch_2015}. According to FCC report in 2002, these increasing demands are straining the longstanding and outmoded spectrum policies \cite{Federal2002Spectrum}. Such a demand disparity is ultimately tied to the end-user behavioral pattern which consistently shows that end-users are very active during the peak hour, and idle at the off-peak time \cite{Song2010Limits,eagle2006reality,farrahi2008discovering}. Interestingly, the time varying end-user activities, that are ultimately behind this mismatch, can be exploited to solve this demand disparity.

As the peak-to-average demand problem has become more severe, the number of researches on efficient control of the supplied services to best match the demand patterns has grown remarkably. \emph{Dynamic pricing} approach has gained a considerable amount of attention as a natural means of influencing the users' economic responsiveness to reduce their peak demand \cite{dalkilic2014pricing}. In fact, some network operators outside the USA (such as Orange, MTN and Uninor) have already started using adaptive pricing schemes to mitigate the excessive cost resulting from high bandwidth consumption \cite{Economist2009Mother}-\cite{Higginbotham2010Mobile}. Moreover, the author in \cite{Ha2012Pricing} shows how dynamic, time-dependent usage pricing incentivizes end-users to spread out their bandwidth consumption more evenly across different times of the day helping ISPs to overcome the problem of peak congestion. There has been other attempts to formulate and study such a dynamic assignment of service prices in \cite{Ha2012Pricing}, whereby the economic responsiveness of the end-users has been utilized through discount incentives to defer the peak network load towards the off-peak time. These schemes, however, rely on the assumption that end-users are willing to change their activity profile based on the network prices, which seems a rather strong assumption since end-users activities are tied to non-flexible constraints such as work, school, and sleep times. Thus, there is an inherent need for a different paradigm of smart pricing that does not only facilitate the best utilization of the network resources, but reasonably accounts for the user's ability and willingness to comply with certain incentive strategies.

WiFi data offloading, on the other hand, has a significant promise towards reducing the peak load of cellular networks. There has been a growing interest in analyzing and optimizing different offloading schemes that \emph{reactively} reschedule the cellular traffic over the WiFi infrastructure while incentivizing end-users through less payments and higher data rates \cite{dimatteo2011Cellular,xiaofeng2013Offloading,mimran2013android}. Attempts to analyze the economics of such WiFi offloading approach are discussed in \cite{lee2013Economics,Lin2013Economics}. In \cite{lee2013Economics}, the scenario where end-users with expected WiFi connectivity utilize \emph{delayed WiFi download} strategies to minimize their payments in a Stackelberg game has been considered, with cellular carriers play first by setting service prices to maximize their profit. It has been shown that win-win situations can be achieved under some conditions on the WiFi availability, end-user willingness to pay, and demand characteristics. In \cite{Lin2013Economics}, the competition between WiFi Access Point (AP) operators and cellular carriers has been considered, where equilibrium points of  WiFi-offloading amounts and corresponding pricing have been characterized. Those model, however, tends to depend mainly on WiFi availability which affects end-user payment and carrier profit and they also did not account for WiFi accessibility cost.

In a recent work \cite{Tadrous2013Pricing}, a different approach for price usage in data networks was proposed. In particular, pricing has been concurrently employed with proactive data services to minimize the time average cost incurred by the carriers. The primal use of pricing incentives has been to enhance the certainty about which service the end-user will ask for. Hence, pricing has not conflicted with the end-user activity. The high certainty about the future demand allowed for an efficient employment of \emph{proactive} data services, a scheme through which predictable peak-hour demand is served \emph{ahead of time} (mainly during the off-peak hour) to smooth-out the network loads and minimize the expected carrier costs. In that work, however, the impact of pricing allocation on the expected profit of the carrier has not been studied.

Device-to-device (D2D) communication has been proposed in \cite{yu2011resource} as a promising technology that can relief the wireless networks congestion. A pair of end-user moving within a close proximity to each other, can establish a D2D link that can be operated in the unlicensed spectrum band, such as the industrial, scientific and medical (ISM) radio bands. These D2D links when used as a traffic offloading approach introduces very little or no monetary cost for end-users. Through exploring the social ties and influence among individuals, \cite{zheng2014social} studied D2D traffic offloading in social networks where a novel approach to improve the performance of D2D communications was proposed as a secondary link over cellular system. The authors in \cite{al2011game} studied the scenario of traffic offloading for common contents to a group of end-users who are equipped with D2D communication. The optimal solution was achieved when the carrier sends different portions of content to some end-users who share them to other end-users.

The main contribution of this paper is to develop a new framework that enables a wireless carrier to increase its profit while end-users reduce their payments. This goal is achieved by exploiting D2D communications when the carrier's smart pricing coupled with end-users proactive caching is applied so as to create a marketplace where all players act individually and strategically. In this work we introduce a carrier assisted risk-sharing marketplace model that allows multiple non-cooperative end-users to trade multiple distinct contents while minimizing each end-user's payment. In this proposed marketplace, the carrier applies dynamic commission on each trading process and harnesses smart pricing schemes to control the amount and timing of such downloads in a way that reduces its costs, and yields higher profit. End-users, on the other hand, leverage \emph{proactive} data downloads over cheap network interfaces and a possibility of trading them to minimize their own payments. We formulate the tension between the two parties as a Stackelberg game where the carrier sets prices first and $N$ end-users respond with proactive downloads and their selling prices. We take into account the uncertainties of connectivity between users, end-users demand, and the freshness of proactive downloads.

We characterize the resulting equilibria points and compare them with the baseline equilibrium defined under flat pricing and no proactive downloads. We show that, caching some data contents in off-peak time enables end-users to save more money and hence increase their chance of reducing their payment below the baseline. This proactive caching shifts some of the peak time load to the off-peak time and hence reduces the service cost for the carrier. This allows the carrier to increase its profit above the baseline. Furthermore, We show that a carrier allowing proactive downloads with trading can make more profit than the case of allowing proactive downloads only without trading. Users, on the same vein, save more by trading their proactive downloads during the peak time and reduce their expected payment.

The rest of this paper is organized as follows. In Section \ref{sec:sys_mod}, we layout the system setup and define the characteristics of its main components. We formulate the Stackelberg game and study the structure of its equilibrium in Section \ref{sec:prob_form}. In Sections \ref{sec:disconnected_users} we characterize the win-win situations for disconnected end-users while in Section \ref{sec:two_connected_users} we characterize the win-win situation for the connected end-users. Numerical results are provided in Section \ref{sec:res}, and the paper is concluded in Section \ref{sec:conc}.

% ===================================================================
% =========================== System Model ==========================
% ===================================================================

\section{System Model} \label{sec:sys_mod}

In this paper, we consider a wireless communication system consisting of a single \textit{wireless network carrier} who supplies $M$ data contents to $N$ end-users upon demand as shown in Fig. \ref{fig:sys_model}. Depending on end-users consistent activities pattern, the carrier divides the day into \textit{off-peak} and \textit{peak} times. The load requested by end-user $j$ at time $t$ is denoted by $L_t^{(j)}$, where $t\in \{o,p\}$ with $o$ and $p$ representing the off-peak and peak times, respectively. These loads are non-negative random variables and assumed to satisfy $\hat{L_p}^{(j)}>\hat{L_o}^{(j)}$ where $\hat{L_t}^{(j)}=\mathbb{E}[L_t^{(j)}]$. In order to reduce its peak service cost, the carrier holds a marketplace that allows end-users to trade pre-cached data contents between each other where the carrier earns a commission from each trade. In the following, we introduce \textit{contents proactive download}, \textit{contents trading} and \textit{pricing} models for all parties.

\begin{figure}[h!]
  \centering
  \includegraphics[width=0.35\textwidth]{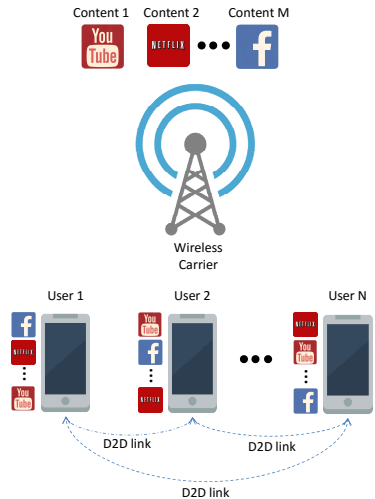}\\
  \caption{An illustration of the system model}\label{fig:sys_model}
\end{figure}

\subsection{Smart Pricing Policy}
By harnessing end-users connectivity to each other and their expected consumption patterns, the carrier produces a dynamic pricing policy $(y_p, y_o, \gamma)$ which denotes the peak time price, off-peak time price, and the trading commission, respectively. This smart pricing strategy is aimed to spread-out the network peak time load and maximizes the carrier's expected profit. Due to inherent market competition and the possibility of regulator rules, the smart pricing policy is assumed to lie within a constraint set $\mathcal{Y}$. In our analysis, we focus on maximum price constraints. In particular, $\cY:=\{0 \leq y_t\leq \hy, t\in\{o,p\} \}$, where $\hat{y}$ is an upper limit price imposed by the regulator. Furthermore, the commission is defined such that $\gamma\in{[0,1]}$. knowingly, the potential proactive download and the trading actions affects the carrier pricing policies, and hence the carrier optimizes its smart pricing policy so as to influence the marketplace towards its profit maximization.

\subsection{Proactive Content Download}
Inspired by the growing body of evidence that human behavior is highly predictable \cite{Song2010Limits,eagle2006reality,farrahi2008discovering,malandrino2012proactive}, we assume that end-users can learn and predict their own and other end-users future consumption and consequently apply proactive content download decisions during the off-peak time taking on consideration all players interest to minimize their individual daily expected payments by utilizing cheap network prices. end-user $j$'s peak time load is defined to be the sum of requests to $M$ uncorrelated data contents (e.g., YouTube video, Netflix movie, CNN news thread, etc), that is, $L_p^{(j)} :=\sum_{m=1}^M{S_m \I_{m}^{(j)}}$, where $S_m$ is the size of data content $m$, and $\I_{m}^{(j)}$ is the indicator that end-user $j$ requests data content $m$ in the peak time. We use $p_m^{(j)}$ to denote $P(\I_{m}^{(j)}=1)$ and we call it here user $j$'s \emph{interest} in content $m$ during the peak time and assume it can be predicted by all parties. We also assume that all end-users take the decision to download data contents at the beginning of each time slot $t \in \{o,p\}$.
%Inspired by the growing body of evidence that human behavior is highly predictable \cite{Song2010Limits,eagle2006reality,farrahi2008discovering}, we assume that end-users can learn and predict their future consumption. Consequently, they apply proactive content download decisions during the off-peak time to minimize their daily expected payments by utilizing cheap network prices. end-user $j$'s peak time load is defined to be the sum of requests to $M$ uncorrelated data contents (e.g., YouTube video, Netflix movie, CNN news thread, etc), that is, $L_p^{(j)} :=\sum_{m=1}^M{S_m \I_{m}^{(j)}}$, where $S_m$ is the size of data content $m$, and $\I_{m}^{(j)}$ is the indicator that end-user $j$ requests data content $m$ in the peak time. We use $p_m^{(j)}$ to denote $P(\I_{m}^{(j)}=1)$ and we call it here end-user $j$'s \emph{interest} on content $m$ during the peak time and assume it is a-priori known to all parties. We also assume that all end-users take the decision to download data contents at the beginning of each time slot $t \in \{o,p\}$.
Driven by the different pricing strategies employed by the carrier and the selling prices of the other users, the proactive download decision made by end-user $j$ in the off-peak time for data content $m$ is denoted by $x_m^{(j)}$, where $x_m^{(j)}\in\left[0,S_m\right]$. Essentially, each end-user aims to maximally benefit from the potentially lower off-peak prices by pre-caching data content for his future demand.

We also assume that the carrier cooperates with end-users and inform them with the availability of data contents at other end-users. Consequently, end-users can save some money when they purchase these data contents with lower prices from any end-user offering them for sale in the peak time. Moreover, end-users who bought these data contents at the off-peak time can compensate some of their payments by trading their proactive downloads. On the other hand, carrier can achieve some gain through the commissions it earns from the content trading actions of the end-users. Privacy concerns always arise in any social based system model. Although it is non of this paper's purposes to suggest solutions to solve these issues, one way to tackle these issue is as suggested in \cite{zhang2015social} by designing selectivity control mechanisms that allow users to choose what contents they are willing to share and what contents they are not.

\subsection{Content Trading}
In order to allow content trading, we assume that end-users are equipped with WiFi direct devices which allow the creation of peer-to-peer connections between WiFi client devices without requiring the presence of a traditional WiFi infrastructure network (i.e., AP or router). In this article we focus on the case of one-hop D2D communications only, and show how it can be applied to allow end-users to trade data between them and get some saving in their payments. We use the indicator $\I_w^{(j,i)}$ to indicate that end-user $j$ and end-user $i$ are connected together and user $j$ is offering the lowest selling price to user $i$. Thus, $\omega_p^{(j,i)} = P(\I_w^{(j,i)} = 1)$ is the probability that end-user $j$ can transfer data to end-user $i$ during the peak time. We will be calling $\omega_p^{(j,i)}$ the connectivity of end-user $j$ to end-user $i$.

We assume that source coding is applied so that the size, beginning and end of each data content are known to all parties. Moreover, data contents get updated in the peak time and the carrier always has the most fresh version of each content. The carrier keeps records for what end-users download and when they download it. This allows the carrier to know which part of the downloaded content at each end-user is still fresh. The proactive download carried out by the end-user in the off-peak time is not guaranteed to be entirely fresh at the time of consumption. That is, some content (especially news) may receive fast updates that essentially render part of the proactive download irrelevant for peak time consumption. We capture the freshness/correlation of data content $m$ by a factor $\alpha_m\in{[0,1]}$ which is the fraction of $x_m^{(j)}$ that is fresh for consumption in the peak time.

To compensate his payment, end-users $j$ can sell the fresh portions of his proactive download to other end-users through the peer-to-peer connection during the peak-time by setting a per-unit data price $y_s^{(j)}$. In order to avoid contents overlapping, we have $\I_{x_m}^{(i,j)}$ as an indicator that end-user $i$ is going to buy a portion of the proactive download of end-user $j$ and this happens when $x_m^{(i)}\leq x_m^{(j)}$. We use $\chi_m^{(i,j)}$ to denote $P(\I_{x_m}^{(i,j)}=1)$.

% ===================================================================
% ======================== Problem Formulation ======================
% ===================================================================

\section{Problem Formulation} \label{sec:prob_form}
In this section, we first formulate the problem as a mutli-stage Stackelberg game. Then, the Stackelberg equilibrium of the proposed game is investigated.

\subsection{Stackelberg Game Formulation}
We formulate the problem as a Stackelberg game with a complete and perfect information. In this game, the carrier is the leader who starts by announcing his optimum smart pricing strategy ($y_p^*, y_o^*, \gamma^*$) taking into consideration the expected end-users reaction. After receiving the carrier pricing strategy, each end-user announces his best selling price $y_s^{*(j)}$. Finally, and based on the received selling prices, all end-users react and choose their optimum proactive decision $x_m^{*(j)}$ for each data content $m$. In the following two subsections we formulate the cost function for the end-users and the  carrier, respectively.
\subsubsection{End-user Optimization}
Driven by the announced carrier prices and the a-priori known end-users interests, end-user $j$ takes his proactive decisions in the off-peak time having in mind the possibility of selling them to the other $N-1$ end-users in the peak time. Moreover, for the contents that end-user $j$ decides to discard caching in the off-peak time, he tries to buy the fresh portion of it in the peak time from any other end-user who cached it. Now, since the carrier has the most fresh version of these data contents, during the peak time, end-user $j$ refreshes his proactive downloads and the contents he purchased from the other end-users. Hence, the expected payment of end-user $j$ is given by:
\begin{equation} \label{Eq_UserNPayment}
\begin{split}
\mu^{(j)} &= y_o \biggl( \hat{L_o}^{(j)} + \sum_{m=1}^M x_m^{(j)} \biggr) + \sum_{m=1}^M \sum_{\substack{i=1 \\ i\neq j}}^N y_s^{(i)} \omega_p^{(i,j)} ( x_m^{(i)} - x_m^{(j)}) \chi_m^{(j,i)} p_m^{(j)}\alpha_m\\
& + y_p \sum_{m=1}^M \sum_{\substack{i=1 \\ i\neq j}}^N \omega_p^{(i,j)} \biggl[ S_m - \alpha_m \biggl(x_m^{(j)}+(x_m^{(i)}- x_m^{(j)}) \chi_m^{(j,i)}\biggr) \biggr] p_m^{(j)}\\
& + y_p \prod_{\substack{i=1 \\ i\neq j}}^N \biggl(1- \omega_p^{(i,j)} \biggr) \sum_{m=1}^M (S_m - \alpha_m x_m^{(j)}) p_m^{(j)} - y_s^{(j)} \sum_{m=1}^M \sum_{\substack{i=1 \\ i\neq j}}^N \omega_p^{(j,i)}(1-\gamma) \alpha_m ( x_m^{(j)} - x_m^{(i)}) \chi_m^{(i,j)} p_m^{(i)}
\end{split}
\end{equation}

In (\ref{Eq_UserNPayment}), we note that end-user $j$ is minimizing his payment by caching proactively an amount $x_m^{(j)}$ from each data content $m$ during the off-peak time using the cellular network resources. He also purchases the fresh non-overlapping portion of other end-users proactive downloads when they are connected to him. Moreover, he tries to sell his proactive downloads in the peak time to any other end-user who is looking for it. So, he tries to find the best selling price to offer his proactive download for sale. Thus, end-user $j$'s best-response is given by:
\begin{equation}\label{Eq_OptJStrategy}
(x_m^{*(j)}, y_s^{*(j)}) := \underset{\substack{x_m^{(j)} \in{[0,S_m]} \\ 0 \leq y_s^{(j)} \leq y_p}} \argmin \mu^{(j)}
\end{equation}

\subsubsection{Carrier Optimization}

In the Stackelberg formulation, the carrier plays first by finding a new smart pricing that maximizes his expected profit $R$. This profit $R$ is shaped by the difference between its incurred cost to provide the service, denoted by $\eta$, and the received payments from all users, i.e.,
\begin{equation} \label{Eq_NSPProfit}
R := \sum_{j=1}^N \mu^{*(j)} - \eta
\end{equation}
where $\mu^{*(j)}$ is the expected optimized payment for end-user $j$. It is important to mention that the total payments received for all end-users (i.e. $\sum_{j=1}^N \mu^{*(j)}$) includes the commission from every trade. We consider the carrier's cost to be proportional to the maximum load it incurs from all end-users throughout the day which is the sum of the off-peak and peak loads. The off-peak load represents the sum of the end-users proactive downloads and the average off-peak demand. The peak load represents the refreshing and completion of all data contents for all end-users. Thus, the carrier's cost is given by:
\begin{equation} \label{Eq_NSPCost}
\begin{split}
\eta &= \beta \mathbb{E} \left[ \max \left\{ \sum_{j=1}^N \biggr [\hat{L_o}^{(j)} + \sum_{m=1}^M x_m^{*(j)} \biggr], \right. \right. \sum_{j=1}^N \sum_{\substack{i=1 \\ i\neq j}}^N \sum_{m=1}^M \I_w^{(j,i)} \biggl[ S_m - \alpha_m \biggl( x_m^{(j)} + (x_m^{(i)}- x_m^{(j)})\I_{x_m}^{(j,i)} \biggr) \biggr] \I_{m}^{(j)}\\
& \left. \left. + \sum_{j=1}^N \prod_{\substack{i=1 \\ i\neq j}}^N \biggl( 1- \I_w^{(j,i)} \biggr) \sum_{m=1}^M (S_m - \alpha_m x_m^{(j)} )  \I_{m}^{(j)} \right\} \right]
\end{split}
\end{equation}
where $\beta$ is the cost factor of the carrier. The carrier aims to find the prices $y_o^*,y_p^*$ and the commission $\gamma^*$ that maximizes his profit. Hence, taking into consideration all users, the optimal pricing policy is:
\begin{equation} \label{Eq_SPOpt}
(y_o^*,y_p^*,\gamma^*) := \underset{\substack{(y_o,y_p) \in \cY \\ \gamma \in [0,1] }} \argmax R
\end{equation}

We denote the equilibrium profit under optimized smart pricing by $R^*$. The Stackelberg equilibrium for the proposed game is investigated in the following subsection.

\subsection{Stackelberg equilibrium}
Generally, the Stackelberg Equilibrium (SE) for a Stackelberg game can be obtained by finding its subgame perfect Nash Equilibrium (NE). As we can see from the game formulation, all end-users are competing with each other and creating a non-cooperative subgame.
%oreover, this non-cooperative subgame is, in fact, a symmetric game as all players have identical strategy spaces and $u_i(s_i,s_{-i}) = u_j(s_j,s_{-j})$ for $s_i=s_j$ and $s_{-i}=s_{-j}$ for all $i,j \in\{1,2,\cdots,N\}$. Thus we can write $u(t,s)$ for the utility to any player playing strategy $t$ in profile $s$ and denote a symmetric game by the tuple $[N, S, u()]$.  We refer to a strategy profile with all players playing the same strategy as a symmetric profile, or, if such a profile is a Nash equilibrium, a symmetric equilibrium.
For this non-cooperative subgame, NE is defined as the the strategy profile at which no player can minimize his payment by changing his strategy unilaterally.
%Formally, the Nash equilibrium is defined as follows:
%\begin{definition}
%For the proposed non-cooperative sub-game with a normal form $G = \Big\{ \mathcal{N}, \{\mathcal{A}_j\}_{j \in \mathcal{N}}, \{U_j\}_{j\in\mathcal{N}} \Big\}$ where $U_j = - \mu^{(j)}$ as defined in (\ref{Eq_UserNPayment}). Let $\textbf{x}^*$ be a $M\times N$ array representing each user's best proactive download decisions for all the $M$ contents. Also let $\textbf{y}^*$ be a $1\times N$ represents each user's best selling price. $\textbf{x}^*$ and $\textbf{y}^*$ are said to attain Nash equilibrium (NE), if and only if, they satisfy the following set of inequalities:
%\begin{equation}\label{NE}
%  U_j(x_{j}^{*},{y^{*}}_{j}, \textbf{x}_{-j}^{*},{\textbf{y}^{*}}_{-j}) \geq   U_j(x_{j},{y}_{j}, \textbf{x}_{-j}^{*},{\textbf{y}^{*}}_{-j}) \ \ \  \forall x_{j}, y_{j} \in \mathcal{A}_j, \ \ \ \ j \in \mathcal{N}
%\end{equation}
%\end{definition}

Now, since there is only one carrier, the best response of the carrier can be obtained by solving (\ref{Eq_SPOpt}). However, because the carrier derive its pricing strategy based on the expected end-users payments, backward induction will be followed by finding the optimum solution for the end-users first, then we find the best pricing strategy for the carrier.
We now show the existence and uniqueness of the subgame perfect Nash Equilibrium as defined in \cite{han2012game}.

\begin{proposition}\label{proposition_2}
In the studied single leader multi-followers game, given that the carrier has announced his best response, there exists a unique NE for the non-cooperative subgame between the N end-users (followers).
\end{proposition}
\begin{proof}
%we start the proof by this lemma:

%\begin{lemma}\label{lemma_SE}
%In a symmetric game $[N, S, u()]$ with S nonempty, compact, and convex and with $u(s_i,s_1,\cdots, s_N)$ continuous in $(s_i,s_1,\cdots, s_N)$ and quasi-concave in $s_i$
%, the best response correspondence $b(s) = \argmax_{t\in S} u(t, s)$
% is nonempty, convex-valued, and upper hemicontinuous.
%\end{lemma}
%\begin{proof}
%Since $u(S, s)$ is the continuous image of the compact set S, it is compact and has a maximum and so b(s) is nonempty. $b(S)$ is convex because the set of maxima of a quasiconcave function $(u(·, s))$ on a convex set $S$ is convex. To show that $b(·)$ is upper hemicontinuous, notice that for any sequence $s^n \Leftrightarrow s$ such that $s^n \in b(s^n)$ for all $n$, we have $s\in b(s)$. To see this, note that for all n, $u(s_i^n, s^n) \geq  u(s_i^′, s^n)$ for all $s_i^ \in S$. So by continuity of $u(·)$, we have $u(s_i,s) \geq u(s_i^′,s)$
%\end{proof}
%Now, $b(·)$ is a correspondence from the nonempty, convex, compact set $S$ to itself. By Lemma \ref{lemma_SE}, $b(·)$ is a nonempty convex-valued, and upper hemicontinous correspondence.
%Thus, by Kakutani’s Theorem there exists a fixed point of $b(·)$ which implies that there exists an $s\in S$ such that $s\in b(s)$ and so all playing s is a symmetric equilibrium.

The existence follows from Schauder Fixed-Point Theorem \cite{istrc1981fixed} since the best-response functions of the end-users are continuous and end-users strategy sets are convex and compact. For the uniqueness, we can see that for a fixed leader strategy the sum of end-users utilities is diagonally strictly concave. Thus, by Rosen theorem, we have a unique Stackelberg Equilibrium \cite{rosen1965existence}.
\end{proof}

Now, we study the effect of end-users proactive behaviour and investigate its equilibrium point and then compare it with the flat pricing scenario. Then we study the effect of trading between end-users and investigate its equilibrium point and then compare it with the flat pricing scenario.

% ===================================================================
% ========================= Disconnected end-users Game ========================
% ===================================================================

\section{Disconnected end-users Stackelberg Game}\label{sec:disconnected_users}

In this section, we study the case where end-users are disconnected and playing separately with the carrier (i.e. $\omega_p^{(i,j)}=0$ $\forall i,j$). First, we pose the price and proactive download allocation problem, then we introduce our baseline decisions and define the performance metrics.

\subsection{The Stackelberg Game}

\subsubsection{\bf End-user Optimization}

For the case when $\omega_p^{(i,j)}=0$ $\forall i,j$ , (\ref{Eq_UserNPayment}) reduces to cover the proactive download carried out by the end-user in the off-peak time and the refreshing process that happens in the peak time. Hence the expected end-user payment is given by:
\begin{equation} \label{Eq_User1Payment}
\begin{split}
 \mu^{(j)} = y_o \biggl( \hat{L_o}^{(j)} + \sum_{m=1}^M x_m^{(j)} \biggr)
 + y_p \sum_{m=1}^M(S_m - \alpha_m x_m^{(j)}) p_m^{(j)}
\end{split}
\end{equation}

Besides the off-peak demand, each end-user downloads proactively an amount of $x_m^{(j)}$ from data content $m$ during the off-peak time. Then he refreshes his proactive downloads and completes the remaining part of each data content during the peak-time. The user's optimization on such decisions is given by: $\min_{\{x_m\in[0,S_m]\}_m} \mu^{(j)}$, which ultimately yields the following structure to optimal proactive download response over the cellular resources, $x_m^{*(j)}$:
\begin{equation}
\label{eq:x^*}
x_m^{*(j)}=\begin{cases} S_m, & \pi_m^{(j)}\geq\frac{y_o}{y_p},\\
0, & \text{otherwise}
\end{cases}
\end{equation}
where $\pi_m^{(j)}:=\alpha_m p_m^{(j)}$ captures the \emph{freshness-popularity product} of source $m$ for end-user $j$. We can see that sources with high values of $\pi_m^{(j)}$ are more likely to be proactively served. Thus, without loss of generality, let us assume in the sequel that data sources are ordered such that $\pi_1^{(j)}> \cdots > \pi_M^{(j)}$, and we use $\pi_0^{(j)}:=1$. Then Let us combine the popularity-freshness products of all end-users for all data contents together in one collection and sort them descendingly to get a new set  $\Pi=\{\pi_k,k\in\{0,1,2,\cdots,MN\}\Big\vert\pi_0>\pi_1>\pi_2>\cdots>\pi_{MN}\}$. We also denote by $\mu^{*(j)}$ the optimized end-user payment through off-peak proactive downloads.

\subsubsection{\bf Carrier Optimization}
As there is no trading between users, because they are not connected together, the carrier will not earn any commission as a revenue resource and then (\ref{Eq_NSPCost}) reduces as follows:
\begin{equation} \label{Eq_User1SPCost}
\begin{split}
\eta = \beta \mathbb{E} \left[ \max \left\{ \sum_{j=1}^N \biggr( \hat{L_o}^{(j)} + \sum_{m=1}^M x_m^{*(j)} \biggr), \right. \right.
 \left. \left. + \sum_{j=1}^N (S_m - \alpha_m x_m^{*(j)} ) \I_{m}^{(j)} \right\} \right]
\end{split}
\end{equation}

Hence, the optimal pricing policy is $(y_o^*,y_p^*):=\arg\max_{(y_o,y_p)\in\cY} R$, which implicitly involves the assignment of a threshold $M^{*(j)}$ for which end-user $j$ will proactively fetch content from sources with indexes $m\leq M^{*(j)}$ over the cellular connection. We denote the equilibrium profit under optimized smart pricing by $R^*$.

\subsection{Performance Metrics}
To evaluate the potential gains of smart pricing and its effect on influencing the demand to minimize the cost, we compare the resulting equilibrium point $(\mu^{*(j)}, R^*)$ with a \emph{baseline} point $(\tilde{\mu}^{(j)}, \tilde{R})$. The baseline is defined under the scenario of flat pricing $\hy$ over the whole day, and no proactive downloads. This is motivated by the flat pricing strategies employed currently by several carriers such as AT\&T's \$10/GB, or \$5/50MB.  Further, we adopt the {\bf savings gain} $\Delta \mu^{(j)}:=\tilde{\mu}^{(j)} - \mu^{*(j)}$, and {\bf profit gain} $\Delta R:=R^*-\tilde{R}$ as our performance metrics for this system. We investigate the potential of win-win situations, i.e., $\Delta \mu^{(j)}>0 \hspace{1mm} \forall j \in \{1,2,\cdots,N\}$ and $\Delta R>0$ under the maximum price constraints.

Under maximum price constrains, the carrier's peak price is $y_p = \hat{y}$ while its off-peak price is always $y_o \leq \hat{y}$. Based on the carrier choice of $y_o$, each end-user will react according to (\ref{eq:x^*}) and cache $M^{*(j)}$ data contents proactively. The carrier can simply set $y_o = \hat{y} \pi_k$ for $k\in\{0,1,2,\cdots,MN\}$, i.e. $y_o$ takes values from $y_o=\hat{y} \pi_0=\hat{y}$ to $y_o=\hat{y} \pi_{MN}$. Each value of $k$ has a set of corresponding $k^{(j)}$'s where every $k^{(j)}$ is related to end-user $j$ and $k=\sum_{j=1}^{N} k^{(j)}$. This means that end-user $j$ will cache contents from sources with indexes $m \leq k^{(j)}$. Consequently, the optimum off-peak price is $y_o=\hat{y}\pi_{M^*}$ where $M^*$ has as set of corresponding $M^{*(j)}$'s with $M^{*(j)}$ related to end-user $j$ and $M^*=\sum_{j=1}^{N} M^{*(j)}$.

\begin{proposition} \label{prop:mp}
The equilibrium pricing strategy under maximum price constraints is given by: $y_p^*=\hy$, and $y_o^*=\hy \pi_{M^*}$, where $M^* = $
\begin{equation}
\begin{split}
\underset{k\in\{0,\cdots,MN\}}{\arg\max} \hy\pi_{k}  \sum_{j=1}^{N}\left(\bL_o^{(j)}+\sum_{m=1}^{k^{(j)}}S_m\right)+\hy\sum_{j=1}^{N}\left(\sum_{m=1}^{k^{(j)}}S_m(1-\al_m^{(j)})p_m^{(j)} +\sum_{m=k^{(j)}+1}^{M}S_m p_m^{(j)}\right)\\
- \beta\E\left[\max\left\{\sum_{j=1}^{N}\left(L_o^{(j)}+\sum_{m=1}^{k^{(j)}}S_m\right),\sum_{j=1}^{N}\left(\sum_{m=1}^{k^{(j)}}S_m(1-\al_m^{(j)})\I_m^{(j)}+\sum_{m=k^{(j)}+1}^{M}S_m\I_m^{(j)}\right) \right\}\right].
\end{split}
\end{equation}

\end{proposition}
\begin{proof}
The proof is straightforward by observing the threshold based structure to the proactive downloads. And the threshold $M^*$ maximizes the carrier's profit.
\end{proof}
We note that $\pi_0^{(j)}=1 \hspace{1mm} \forall j$ by definition, so the carrier can never get lower profit than of flat pricing. Next, we characterize the condition on the win-win equilibrium.
\begin{theorem}\label{Thrm_Disconnected}
Under the assumptions of Proposition \ref{prop:mp}, the win-win equilibrium $\Delta \mu^{(j)}>0 \hspace{1mm} \forall j \in \{1,2,\cdots,N\}$ and $\Delta R>0$ is attained if and only if:
\begin{equation}\label{Eq_WinWin_NUsers_Disconnected}
\begin{split}
\beta \sum_{j=1}^N \sum_{m=1}^M S_m p_m^{(j)} > \beta \left(S_1 +\sum_{j=1}^N \hat{L_o}^{(j)}\right) + \hat{y} (1 - \pi_1^{(1)}) \sum_{j=1}^N \hat{L_o}^{(j)}
\end{split}
\end{equation}
\end{theorem}
\begin{proof}
Proof is provided in Appendix \ref{NM-notrading}.
\end{proof}

Using this theorem the carrier can decide, before starting the game, whether to apply smarting pricing that will allow proactive caching or to use the baseline scenario of a flat pricing by setting $y_o=\hat{y}$. Next we investigate the impact of the connected end-users case on the win-win equilibrium.
% ===================================================================
% ======================== Connected end-users Game =====================
% ===================================================================

\section{Connected end-users Stackelberg Game}\label{sec:two_connected_users}
We now consider the case where end-users are connected together (i.e. $\omega_p^{(i,j)}>0 \forall i,j\in\{1,2,\cdots,N\}$). Hence, there will be always a possibility of trading to happen between them. In the following, we pose the price and proactive download allocation problem, then we introduce our baseline decisions and define the performance metrics.

\subsection{The Stackelberg Game}

\subsubsection{\bf end-users optimization}

end-user $j$ minimizes his net payment by proactively downloading an amount $x_m^{(j)}$ from each source $m$ during the off-peak time as well as selling his proactive downloads to other end-users. He may also purchase some of the other end-users proactive downloads. The optimal proactive download decision $x_m^{*(j)}$ of end-user $j$ is as follows:
\begin{equation} \label{Eq_NOptJProNE}
%\begin{figure*}[t]
%\resizebox{\linewidth}{!}{%
x_m^{*(j)} = \left\{ \begin{array}
{r@{\quad,\quad}l}
S_m&\pi_m^{(j)} > \frac { y_o - \sum\limits_{{\substack{i=1 \\ i\neq j}}}^N y_s^{(j)} (1 - \gamma) \omega_p^{(j,i)} \chi_m^{(i,j)} \pi_m^{(i)} } {\prod\limits_{{\substack{i=1 \\ i\neq j}}}^N \overline{\omega}_p^{(i,j)} y_p + \sum\limits_{{\substack{i=1 \\ i\neq j}}}^N  \omega_p^{(i,j)}\biggr(\overline{\chi}_m^{(j,i)}y_p + \chi_m^{(j,i)} y_s^{(i)} \biggr)}\\
0&otherwise
\end{array} \right.
\end{equation}

where $\overline{\chi}_m^{(j,i)}= 1 - {\chi}_m^{(j,i)}$. As we can see, end-user $j$'s decision depends on his interest in the content $m$ and other end-users' interest. This leads us to the following claim:

\begin{claim}\label{claim_1}
For any content $m$, if $\pi_m^{(j)}> \pi_m^{(i)}, i \neq j$, then $x_m^{*(j)} \ge x_m^{*(i)}$ with equality holds when  $x_m^{*(j)} = 0$.
\end{claim}
\begin{proof}
Consider the case when $w_p^{(i,j)}=1$. Suppose $x_m^{*(j)}=0$ then from (\ref{Eq_NOptJProNE}) $\pi_m^{(j)}\leq\frac{yo}{y_p}$. Since $\pi_m^{(i)}<\pi_m^{(j)}$, then $\pi_m^{(i)}\leq \frac{yo}{y_p}$. This means that $\pi_m^{(i)}<\frac{yo}{y_s^{(j)}}$ which yields to $x_m^{*(i)}=0$.
Now, Suppose $x_m^{*(j)}=S_m$. If $x_m^{*(i)}=0$ then $x_m^{*(j)}>x_m^{*(i)}$. If $x_m^{*(i)}=S_m$ then $x_m^{*(j)}=x_m^{*(i)}$.
\end{proof}

Since we have the assumption that end-users interest are never equal, then there will be always an end-user who has the highest interest in content $m$ and this end-user will be the only seller for this content as he will offer a selling price that will encourage all the other $N-1$ end-users to wait until the peak time and purchase it from him. Moreover, since buyers here are symmetric and no competition exists in the buying process, it is sufficient to consider the case of one buyer and then we can generalize the result to cover the $N-1$ buyers. Therefore, in this section we will limit our study to the case of two end-users (one seller and one buyer) to investigate and show the gain of our proposed model.

We start our analysis by assuming that $\pi_m^{(j)}>\pi_m^{(i)}, \forall m \in \{1,2,\cdots,k_1\}$ and $\pi_m^{(i)}>\pi_m^{(j)}, \forall m \in \{k_1+1,k_1+2,\cdots,M\}$. According to claim \ref{claim_1}, we have 2 cases, the end-user who has a higher interest buys a certain data content while the other end-user waits until the peak-time, or both of them discarding buying in the off-peak time and buy it in the peak time from the carrier.

Hence (\ref{Eq_UserNPayment}) can be rewritten as the following:
\begin{equation} \label{Eq_UserPaymentJ}
\begin{split}
\mu^{(j)} &= y_o \biggl( \hat{L_o}^{(j)} + \sum_{m=1}^{M} x_m^{(j)} \biggr) + y_s^{(i)} \omega_p^{(i,j)} \sum_{m=1}^{M}  \alpha_m ( x_m^{(i)} - x_m^{(j)} ) \chi_m^{(j,i)} p_m^{(j)} \\
& + y_p (1-\omega_p^{(i,j)}) \sum_{m=1}^M (S_m - \alpha_m x_m^{(j)} )  p_m^{(j)}- y_s^{(j)} \omega_p^{(j,i)} (1-\gamma) \sum_{m=1}^{k_1} \alpha_m ( x_m^{(j)} - x_m^{(i)}) \chi_m^{(i,j)} p_m^{(i)}\\
& + y_p \omega_p^{(i,j)} \sum_{m=1}^{M} \biggl[ S_m - \alpha_m \biggl( x_m^{(j)} + (x_m^{(i)}- x_m^{(j)}) \chi_m^{(j,i)} \biggr) \biggr] p_m^{(j)}
\end{split}
\end{equation}

Thus, end-user $j$'s best-response is given by:
\begin{equation}\label{Eq_OptJStrategy_1}
(x_m^{*(j)}, y_s^{*(j)}) := \underset{\substack{x_m^{(j)} \in{[0,S_m]} \\ 0 \leq y_s^{(j)} \leq y_p}} \argmin \mu^{(j)}
\end{equation}

\begin{comment}
In particular, end-user $j$'s best proactive decision will be:
whereas end-user $j$'s best proactive decision for all contents $m \in [1,M]$ is given by:
\begin{equation} \label{General_Eq_xj_opt}
x^{*(j)}_{m} = \left\{ \begin{array}
{r@{\quad,\quad}l}
S_m & \pi^{(j)}_{m} > \frac {y_o - y_s^{(j)} \omega_p^{(j,i)} (1-\gamma) \chi_m^{(i)} \pi_m^{(i)}} {y_s^{(i)} \omega_p^{(i,j)} \chi_m^{(j)} + y_p (1-\omega_p^{(i,j)}) + y_p \omega_p^{(i,j)} ( 1 - \chi_m^{(j)} )}\\
0 & otherwise
\end{array} \right.
\end{equation}
\end{comment}

However, since we have $\pi_m^{(j)}>\pi_m^{(i)}, \forall m \in \{1,2,\cdots,k_1\}$ and $\pi_m^{(i)}>\pi_m^{(j)}, \forall m \in \{k_1+1,k_1+2,\cdots,M\}$, this means that $\chi_m^{(j,i)} = 0$ and $\chi_m^{(i,j)} = 1$ for all $m \in \{1,2,\cdots,k_1\}$ and we have $\chi_m^{(j,i)} = 1$ and $\chi_m^{(i,j)} = 0$ for all $m \in \{k_1+1,k_1+2,\cdots,M\}$. From this, we can simplify (\ref{Eq_NOptJProNE}) as follows:

For all $m \in \{1,2,\cdots,k_1\}$ and given the carrier optimum strategies,  we have:
\begin{equation} \label{Eq_xi_opt_1}
x^{*(i)}_{m} = \left\{ \begin{array}
{r@{\quad,\quad}l}
S_m & \pi^{(i)}_{m} > \frac {y_o} { y_p (1-\omega_p^{(i,j)}) + y_{s}^{(j)} \omega_p^{(i,j)}}\\
0 & otherwise
\end{array} \right.
\end{equation}
\begin{equation} \label{Eq_xj_opt_1}
x^{*(j)}_{m} = \left\{ \begin{array}
{r@{\quad,\quad}l}
S_m & \pi^{(j)}_{m} > \frac {y_o - (1-\gamma) y_{s}^{(j)} \omega_p^{(i,j)} \pi_{m}^{(i)} } { y_p}\\
0 & otherwise
\end{array} \right.
\end{equation}
\\
As we can see from (\ref{Eq_xi_opt_1}) and (\ref{Eq_xj_opt_1}), the proactive decisions of both end-users $j$ and $i$ are affected by the selling price of end-user $j$. Thus, end-user $j$ will evaluate all of the possible selling prices for the sources which he decided to cache and chooses the best one that maximize his revenue. In accordance, end-user $i$ takes his proactive decisions which minimizes his payment too. Hence the optimum flat selling price of end-user $j$ is given by:
\begin{equation}\label{Eq_ysj_opt}
y_{s}^{*(j)} = \frac{y_o-y_p(1-\omega_p^{(i,j)})\pi^{(i)}_{q^{*(j)}}}{\omega_p^{(i,j)}\pi^{(i)}_{q^{*(j)}}} \text{, where } q^{*(j)}=\argmax_{r \in\{1,\cdots,M^{*(j)}\}} \pi^{(i)}_{m} \text{, and } M^{*(j)} \leq k_1
\end{equation}
Similarly, for all $m \in \{k_1+1,k_1+2,\cdots,M\}$ we have:
\begin{equation}\label{Eq_xj_opt_2}
x^{*(j)}_{m} = \left\{ \begin{array}
{r@{\quad,\quad}l}
S_m & \pi^{(j)}_{m} > \frac {y_o} { y_p (1-\omega_p^{(j,i)}) + y_{s}^{(i)} \omega_p^{(j,i)}}\\
0 & otherwise
\end{array} \right.
\end{equation}
\begin{equation}\label{Eq_xi_opt_2}
x^{*(i)}_{m} = \left\{ \begin{array}
{r@{\quad,\quad}l}
S_m & \pi^{(i)}_{m} > \frac {y_o - (1-\gamma) y_{s}^{(i)} \omega_p^{(j,i)} \pi_{m}^{(i)} } { y_p}\\
0 & otherwise
\end{array} \right.
\end{equation}
\\
and the optimum flat selling price of end-user $i$ is given by:
\begin{equation}\label{Eq_ysi_opt}
y_{s}^{*(i)} = \frac{y_o-y_p(1-\omega_p^{(i,j)})\pi^{(j)}_{q^{*(i)}}}{\omega_p^{(i,j)}\pi^{(j)}_{q^{*(i)}}} \text{, where } q^{*(i)}=\argmax_{r \in\{1,\cdots,M^{*(i)}\}} \pi^{(j)}_{m} \text{, and } M^{*(i)} \leq k_1
\end{equation}

%\begin{theorem}\label{coro_1}
%For the case of $N=2$, the less interested end-user $i$ will prefer to buy content $m$ from the most interested end-user $j$ (i.e. $\I_m^{(i,j)}=1$) when
%\begin{equation*}
%\frac{y_o - \pi_m^{(j)}y_p}{(1-\gamma)y_s^{(j)}} < \pi_m^{(i)} < \frac{y_o}{y_s^{(j)}}
%\end{equation*}
%\end{theorem}
%\begin{proof}
%From the definition of $\I_m^{(i,j)}$, we see that $\I_m^{(i,j)}=1$ when $x_m^{(i)} < x_m^{(j)}$. This means that $x_m^{(j)} = S_m$ while $x_m^{(i)} = 0$. A simple manipulation of (\ref{Eq_xi_opt}) and (\ref{Eq_xj_opt}) will lead us to this condition inequality.
%\end{proof} This theorem gives end-user $j$ a guide to which contents he can resell when choosing his selling price given a-priori known interests and the received carrier prices.

\subsubsection{\bf Carrier Optimization}

Now, the carrier's expected profit $R$ is given by:
\begin{equation} \label{Eq_SPProfit}
R := \mu^{*(j)} + \mu^{*(i)} - \eta
\end{equation}
where $\eta$ is defined as in (\ref{Eq_NSPCost}). The carrier aims to find the prices $y_o^*,y_p^*$ and the commission $\gamma^*$ that maximizes his profit as in (\ref{Eq_SPOpt}). This implicitly involves the assignment of a threshold $M^{*(j)} \in \{1,2,\cdots,k_1\}$ for end-user $j$ which defines the data contents he will be able to cache and a similar threshold $M^{*(i)} \in \{k_1+1,k_1+2,\cdots,M\}$ for end-user $i$. Also, carrier's prices $y_o^*$ and $y_p^*$ affects the selling prices of both users $y_s^{*(j)}$ and $y_s^{*(i)}$.

\begin{proposition} \label{Prop_SPEqPolicy}
The carrier equilibrium pricing strategy under maximum price constraints is given by:
\begin{equation}
 (y_p^*, y_o^*, \gamma^*) = \underset{\substack{y_o, y_p \in [0,\hat{y}] \\ \gamma \in [0,1] }} \argmax \left(\mu^{(j)} + \mu^{(j)} - \eta \right)
\end{equation}
\end{proposition}

\subsection{Performance Metrics}
Here, we investigate the potential of the \textit{win-win} situation, i.e., $\Delta{\mu^{(j)}} > 0 \hspace{1mm} \forall j \in \{1,2,\cdots,N\}$ and $\Delta{R} > 0$.
Inspired by the two connected end-users game discussed above, we can extend it to the case where we have $N$ end-users. Now, for each data content $m$, one of the end-users will be the most interested user in this data content $m$ among other end-user. This end-user will download that data content during the off-peak time and then try to trade it in the peak time to the other end-users who did not buy in the off-peak time as shown in claim \ref{claim_1}. Given the diversity of end-users' interests each end-user will proactively cache some data contents and buy some data contents from the other end-users and finally go to the carrier to complete or refresh his downloads. Now, we characterize the conditions for the win-win equilibrium in the following theorem.

\begin{theorem} \label{Thrm_Connected}
Under the assumptions of Proposition \ref{Prop_SPEqPolicy}, the win-win equilibrium $\Delta{\mu^{(j)}} > 0 \hspace{1mm} \forall j \in \{1,2,\cdots,N\}$ and $\Delta{R} > 0$ is attained if and only if
\begin{equation} \label{eq:cond1}
\begin{split}
\beta \sum_{j=1}^{N} \sum_{m=1}^{M} S_m p_m^{(j)} > \beta \left(S_1 + \sum_{\substack{j=1 }}^{N} \hat{L_o}^{(j)} \right) + \hat{y}(1 - \pi_{1}^{(1)})\sum_{\substack{j=1 }}^{N} \hat{L_o}^{(j)} + \hat{y} S_1 \left( 1-\frac{\pi_{1}^{(1)}}{\max\limits_{\substack{j\in \{2,\cdots,N\}}}\pi_{1}^{(j)}} \right) \sum_{j=2}^{N} \pi_1^{(j)}\\
\end{split}
\end{equation}
\end{theorem}
\begin{proof}
proof is provided in Appendix B.
\end{proof}
We can see that the carrier will gain some more profit due to this trading process when it applies the smart pricing policy. On the other hand, end-users will gain some saving from both proactive caching and the trading process.
\begin{proposition} \label{Prop_winwin}
The probability of a win-win situation in the case of the connected users is more than its probability in the case of disconnected users.
\end{proposition}

\begin{proof}
The proof is straightforward by comparing (\ref{eq:cond1}) with (\ref{Eq_WinWin_NUsers_Disconnected}) and noticing that the win-win equilibrium region in the connected users case expands monotonically
\end{proof}

\section{Numerical Results} \label{sec:res}
To verify the derived results, we start with the case of one end-user only with varying interest to study the effect of end-user interest on the system performance. This end-user consumes data from $M=5$ data sources in the peak time. The amount consumed from each data source is $S=100$ units. Sources have different freshness $[1.0, 0.95, 0.90, 0.85, 0.80]$. We start with $p_m=0$ for all data contents and then increase it until we reach $p_m=1$. We take $\hy=1$ and $\beta=0.75$. We plot the profit and saving gains $\Delta R$ and $\Delta \mu$ versus end-user interest. Moreover, we compare the off-peak and peak loads of our model with the flat pricing model. The results are depicted in Fig. \ref{fig:Performance_1User}.

\begin{figure}
	\centering
		\includegraphics[width=0.7\textwidth]{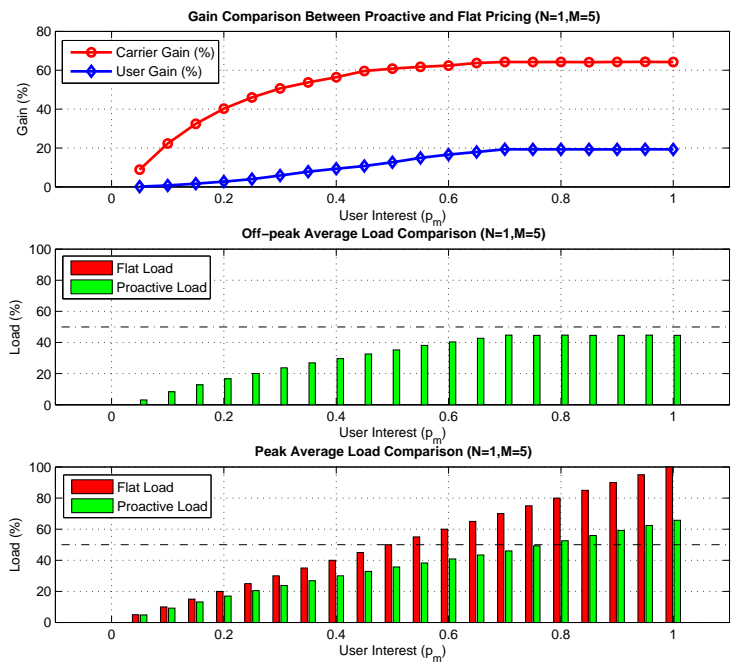}
	\caption{System Performance for ($N=1$) versus end-user interest}
	\label{fig:Performance_1User}
\end{figure}

As we can notice in Fig. \ref{fig:Performance_1User}, a win-win equilibrium is attained whenever end-user interest is positive. This is clear as the smart pricing incentives the end-user to apply proactive download which reduces his average payment. On the other hand, proactive download reduces the carrier peak load which yields reducing the cost. These gains increase as the interest of the end-user increases since he can cache more contents and consequently the peak load decreases. The carrier gains almost $65\%$ while end-user saves about $20\%$ of his payment. Moreover, we can see that allowing end-user to download some of the data contents proactively shifts the peak time load to the off-peak time and hence we have almost flat network utilization.

Next, we study the effect of engaging more end-users in the network. Let us consider a system consisting of $N=1,2,\cdots ,20$ end-users consuming data from the same $M=5$ data sources in the peak time. Those end-users are disconnected and hence they can not trade their proactive downloads. We start by one end-user and then we add more end-users to the network. Here, end-users interests are the same for each content $m$. The freshness of the data contents are $[1.0, 0.95, 0.90, 0.85, 0.80]$. We plot the profit and savings gains $\Delta R$ and $\Delta \mu$ versus number of users. we compare the off-peak and peak loads of our model with the flat pricing model. The results are depicted in Fig. \ref{fig:Performance_NUsers_Uniform}.

\begin{figure}
	\centering
		\includegraphics[width=0.7\textwidth]{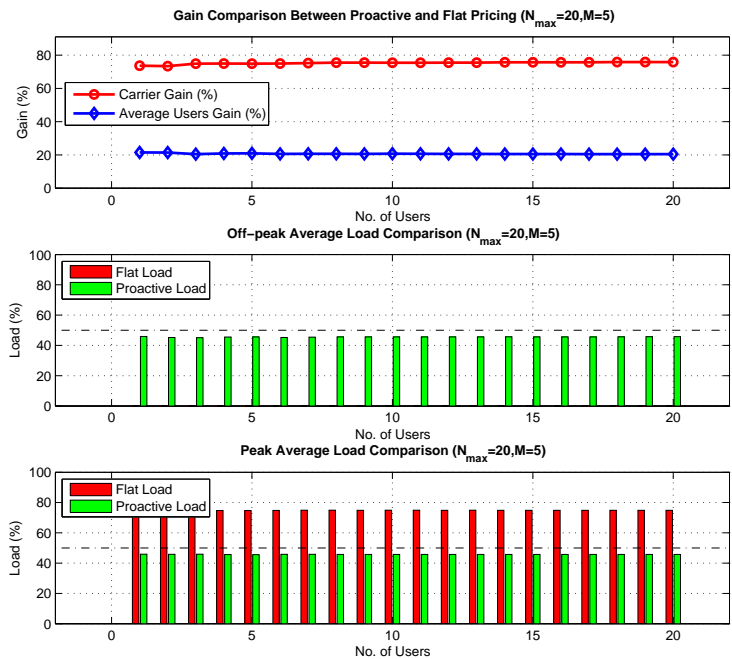}
	\caption{System Performance versus no. of end-users (fixed interests)}
	\label{fig:Performance_NUsers_Uniform}
\end{figure}

As we can notice in Fig. \ref{fig:Performance_NUsers_Uniform}, regardless of the number of end-users the carrier can manage to attain a fixed gain. This is in fact because the carrier smart pricing will not change as it is determined by the highest interested end-user and all end-users have the same interest. So, adding new end-users will not change the carrier prices and hence all end-users will download the same amount proactively. The carrier can attain an equilibrium and achieve almost $75\%$ profit gain. On the other hand, the end-users can achieve an average payment saving of almost $20\%$. Moreover, the network load is distributed between the off-peak and peak time which also relieves the network congestion during the peak time and tackles the problem of network under utilization during the off-peak time.

Now, let us investigate the effect of adding more end-users with increasing interests. We generate an increasing contents interest for the end-users and add them one by one to the network. We keep all the other system parameter as above. We plot the carrier gain versus the number of end-users and we also compare the off-peak and peak loads with the flat pricing model. The results are depicted in Fig. \ref{fig:Performance_NUsers_Increasing}.

\begin{figure}[h!]
	\centering
		\includegraphics[width=0.7\textwidth]{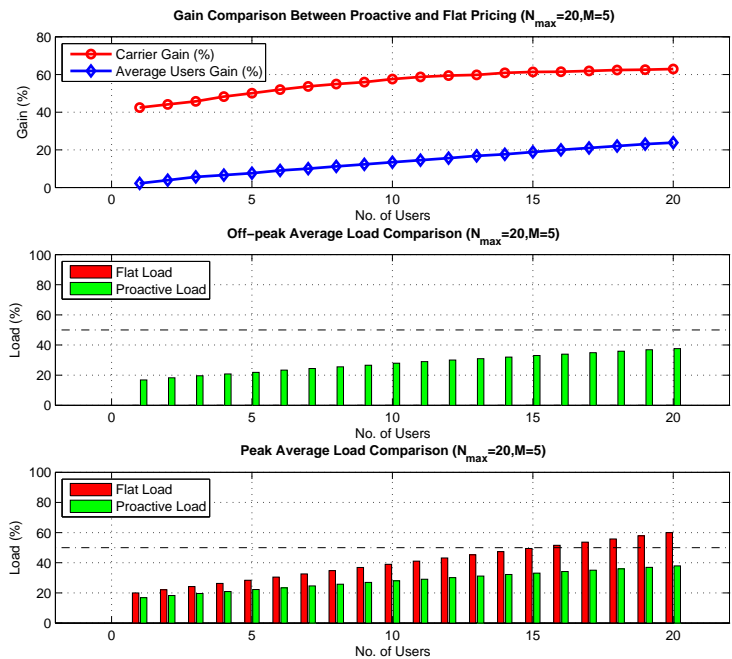}
	\caption{System Performance versus no. of end-users (increasing interest)}
	\label{fig:Performance_NUsers_Increasing}
\end{figure}

As we can notice in Fig. \ref{fig:Performance_NUsers_Increasing}, the carrier achieves more gain when more interested end-users are engaged in the network. That is because adding more interested users makes the carrier increase the off-peak price which means increasing the revenue. Moreover, the average end-users saving increases as the number of more interested end-users increase. This result gives the carrier a signal to invest more in the contents popularity to engage more end-users in the network. Moreover, the network load is distributed between the off-peak and peak time which also relieves the network congestion during the peak time and tackles the problem of network under utilization during the off-peak time.

From Theorem \ref{Thrm_Connected}, we saw that trading between end-users helps both the carrier and the end-users to gain more. To verify this result, let us consider a  system consisting of two end-users consuming data from $M = 5$ data sources in the peak time according to different interests where end-user $j$ is more interested than end-user $i$ as follows: the interest of end-user $j$ is $p_m^{(j)}=[1.0,1.0,1.0,1.0,1.0]$ while the interest of end-user $i$ is $p_m^{(i)}=[0.95,0.90,0.85,0.80,0.75]$. All data contents have the same freshness $\alpha_m =0.90$. The off-peak loads of both end-users are $L_o^{(j)}=0$ and $L_o^{(i)}=0$. The amounts consumed from all sources are equal to $S = 100$ units. We investigate the effect of connectivity between the two end-users on the carrier profit and the end-users gains. We also investigate the effect of end-users connectivity on the off-peak and peak loads. The results are depicted in Fig. \ref{Performance_2Users_Trading}.

\begin{figure}[h!]
  \centering
  \includegraphics[width=0.8\textwidth]{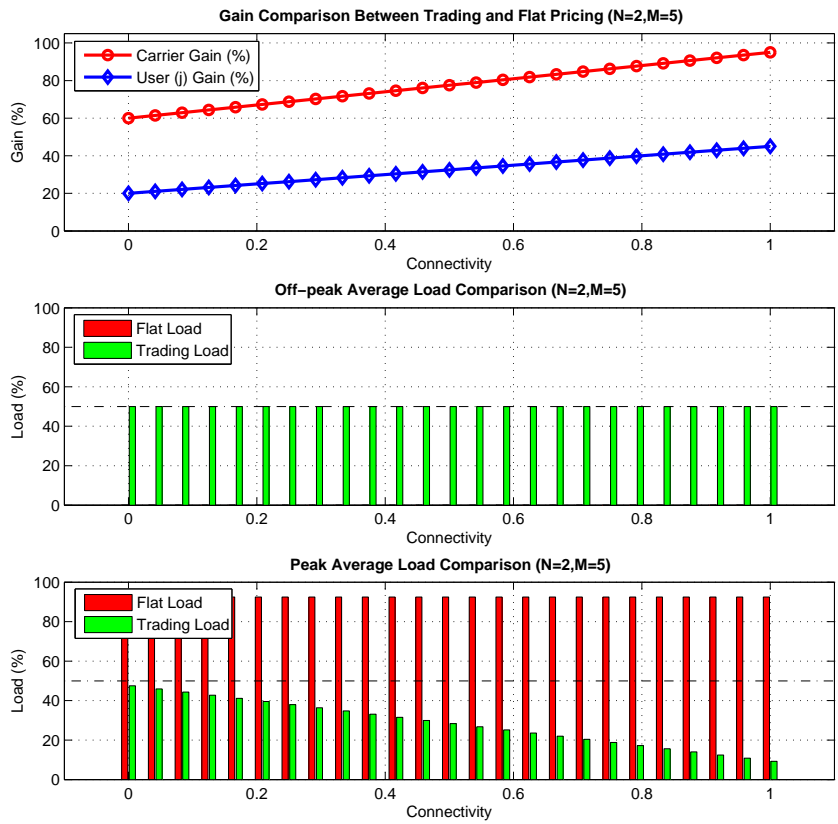}
  \caption{System Performance versus connectivity between end-users}\label{Performance_2Users_Trading}
  \label{fig:Performance_2Users_Trading}
\end{figure}

As we can see, although when the end-users are disconnected (i.e. $\omega_p^{(j,i)} = 0$) they were able to download proactively some of the data contents during the off-peak time and then refresh and download the remaining data contents during the peak time. The carrier achieves a gain of almost $60\%$ while the most interest user (end-user $j$) achieves almost $20\%$. Although the peak load in the flat pricing model hits $95\%$, it is distributed between the off-peak and peak times in the case of our model. The gain of the carrier profit and the saving of end-user $j$ increase as the connectivity increases. On the other hand, when end-users are fully connected (i.e. $\omega_p^{(j,i)} = 1$), the carrier achieves a gain of almost $98\%$ and end-user $j$ achieves almost $45\%$. We can also see that the peak load decreases as the connectivity between end-users increases. The reason behind that is, when users get connected, only end-user $j$ downloads data contents and end-user $i$ wait for him. Moreover, this plot supports the results we mentioned in Proposition \ref{Prop_winwin}. Trading helped the carrier to shift the peak load to the off-peak time and gain more profit whereas end-users saved more when they shared their proactive downloads.

%For end-user $i$, he will not lose or gain and his gain will be fixed as the the connectivity increases. This happens because end-user $j$ dominates the game with the carrier and offers a selling price $y_s^{(j)}=y_p$. This means that end-user $i$ is not saving any of his payment as he is always less interested.

\section{Conclusion} \label{sec:conc}
In this work, we seek to harness the unexploited predictability of users' demand pattern and the possibility of end-users being connected with each other in providing an efficient remedy to the wireless spectrum crisis. In particular, A novel marketplace that is based on risk sharing between end-users is proposed where the tension between carrier and end-users is formulated as a Stackelberg game. In order to maximize its profit, the carrier set a time dependent pricing policy and a varying commission so as to influence end-users to trade contents between each other and balance their demand throughout the day. End-users, on the other hand, utilize their predictable demand and the possibility of them being connected with each other in determining optimal \emph{proactive download} decisions in the off-peak time, and selling price for the peak time, so as to minimize their expected payments. Unlike existing reactive time dependent pricing and WiFi off-loading schemes where end-users have to delay their demand and wait for cheap network prices, this proactive based marketplace offers content before demand and avoid the change of the users' behavior.

We explored the equilibrium points for the case when the D2D is available and when it is not available, and study the impact of the uncertainty of users’ future demands on the system performance. We compared the new equilibria with the baseline scenario of flat pricing model where no caching is allowed. For end-users connectivity, we showed that as the connectivity between end-users increases, all parties achieve more gain. For contents interest, we showed that as the interest of end-users increase, all parties achieve more gain. Moreover, we have shown that having one user at least with a non-zero interest will lead to a win-win situation regardless of the users being connected or not.

\bibliographystyle{IEEEtran}
\bibliography{refs}

\begin{appendices}

\section{Proof of Theorem \ref{Thrm_Disconnected}} \label{NM-notrading}
\begin{proof}
We have $N$ disconnected end-users (i.e. $w_p^{(i,j)}=0 \hspace{1mm} \forall i, \forall j$) consuming data from $M$ data sources. From Proposition \ref{prop:mp}, $y_p = \hat{y}$ and $y_o \leq \hat{y}$. Moreover, let  $\Pi=\{\pi_k,k\in\{0,\cdots,MN\}\Big\vert\pi_0>\pi_1>\pi_2>\cdots>\pi_{MN}\}$ as defined in Section \ref{sec:prob_form}. The carrier, then, can simply set $y_o = \hat{y} \pi_k$ for $k\in\{0,\cdots,MN\}$ and choose the one that gives optimum gain. We now investigate the win-win region $\mathcal{R}_{k}^{d}$, where $d$ denotes disconnected, for each value of $y_o$ where $\mathcal{R}_{k}^{d}$ is defined as:
\begin{equation}
\mathcal{R}_{k}^{d} = \biggl\{ p_m^{(j)}, m \in \{1,\cdots,M\}, j \in \{1,\cdots,N\} \Big\vert y_p = \hat{y}, y_o = \hat{y} \pi_k, \Delta{\mu^{(j)}}> 0 \hspace{1mm} , \Delta{R}> 0 \biggr\}
\end{equation}

\subsection{\textbf{\underline{For $k=0$:}}}
For $k=0$, we have $y_p = \hat{y}, y_o = \hat{y} \pi_0 = \hat{y}$ hence $y_p=y_o$. The saving in end-user payment is:
\begin{equation}\label{Eq_YoYp_Users}
\begin{split}
%\mu^{*(j)} &= y_o \hat{L_o}^{(j)} + y_p \sum_{m=1}^M S_m p_m^{(j)}\\
%\tmu^{(j)} &= y_p \biggl( \hat{L_o}^{(j)} + \sum_{m=1}^{M} S_m p_m^{(j)} \biggr)\\
\Delta \mu^{(j)} &= \tmu^{(j)} - \mu^{*(j)} = (y_p - y_o) \hat{L_o}^{(j)} = 0
\end{split}
\end{equation}
and the carrier profit gain is:
\begin{equation}\label{Eq_YoYp_Carrier}
\begin{split}
%R^*  &= \sum_{j=1}^N \mu^{*(j)} - \eta^* = \sum_{j=1}^N \biggl( y_o \hat{L_o}^{(j)} + y_p \sum_{m=1}^M S_m p_m^{(j)} \biggr) - \beta \sum_{j=1}^N \sum_{m=1}^{M} S_m p_m^{(j)}\\
%\tR &= \sum_{j=1}^N \tmu^{(j)} - \teta = y_p \sum_{j=1}^N \biggl( \hat{L_o}^{(j)} + \sum_{m=1}^{M} S_m p_m^{(j)} \biggr) - \beta \sum_{j=1}^N \sum_{m=1}^{M} S_m p_m^{(j)}\\
\Delta R &= R^* - \tR = (y_o - y_p) \sum_{j=1}^N \hat{L_o}^{(j)} = 0\\
\end{split}
\end{equation}
Thus, the win-win region for this case as follows:
\begin{equation}
\mathcal{R}_{0}^{d} = \biggl\{ p_m^{(j)}, m \in \{1,2,\cdots,M\}, j \in \{1,2,\cdots,N\} \Big\vert y_o = \hat{y}, \Delta{\mu^{(j)}}> 0 \hspace{1mm} \forall j , \Delta{R}> 0\biggr\} = \phi\\
\end{equation}

\subsection{\textbf{\underline{For $k=MN$:}}}
For $k=MN$, we have $y_p = \hat{y}, y_o = \hat{y} \pi_{MN}$. The saving in end-user payment is:
\begin{equation}
\begin{split}
%\mu^{*(j)} &= y_o \biggl( \hat{L_o}^{(j)} + \sum_{m=1}^M S_m \biggr) + y_p \sum_{m=1}^M S_m (1-\alpha_m) p_m^{(j)}\\
%\tmu^{(j)} &= y_p \biggl( \hat{L_o}^{(j)} + \sum_{m=1}^{M} S_m p_m^{(j)} \biggr)\\
\Delta \mu^{(j)} %&= (y_p - y_o) \hat{L_o}^{(j)} + \sum_{m=1}^M S_m (y_p \pi_m^{(j)}-y_o)\\
                 &= y_p(1 - \pi_{MN}) \hat{L_o}^{(j)} + y_p \sum_{m=1}^M S_m (\pi_m^{(j)}-\pi_{MN}) > 0
\end{split}
\end{equation}
and the carrier profit gain can be found by:
\begin{equation}
\begin{split}
%R^*  &= y_o \sum_{j=1}^N \biggl( \hat{L_o}^{(j)} + \sum_{m=1}^M S_m \biggr) + y_p \sum_{j=1}^N \sum_{m=1}^M S_m (1-\alpha_m) p_m^{(j)}- \beta \sum_{j=1}^N \biggl( \hat{L_o}^{(j)} +  \sum_{m=1}^{M} S_m \biggr)\\
%\tR &= y_p \sum_{j=1}^N \biggl( \hat{L_o}^{(j)} + \sum_{m=1}^{M} S_m p_m^{(j)} \biggr) - \beta \sum_{j=1}^N \sum_{m=1}^{M} S_m p_m^{(j)}\\
\Delta R %&= (y_o - y_p - \beta) \sum_{j=1}^N \hat{L_o}^{(j)} + \sum_{j=1}^N \sum_{m=1}^M S_m (y_o-y_p \pi_m^{(j)}) - \beta \sum_{j=1}^N \sum_{m=1}^M S_m (1-p_m^{(j)}) \\
         &= (y_p \pi_{MN} - y_p - \beta) \sum_{j=1}^N \hat{L_o}^{(j)} + y_p \sum_{j=1}^N \sum_{m=1}^M S_m (\pi_{MN} - \pi_m^{(j)}) - \beta \sum_{j=1}^N \sum_{m=1}^M S_m (1-p_m^{(j)}) < 0
\end{split}
\end{equation}
Thus, the win-win region for this case is:
\begin{equation}
\mathcal{R}_{MN}^{d} = \biggl\{ p_m^{(j)}, m \in \{1,2,\cdots,M\}, j \in \{1,2,\cdots,N\} \Big\vert y_o = \hat{y} \pi_{MN}, \Delta{\mu^{(j)}}> 0 \hspace{1mm} \forall j , \Delta{R}> 0\biggr\} = \phi\\
\end{equation}

\subsection{\textbf{\underline{For $0<k<MN$:}}}
For $0<k<MN$, we have $y_p = \hat{y}, y_o = \hat{y} \pi_k$. From (\ref{eq:x^*}) we can see that if $k=1$, only one end-user caches data content $m=1$ and other end-users will not. If we assume that this end-user is end-user $1$, then we have his saving as::
\begin{equation}
\begin{split}
%\mu^{*(1)} &= y_o \biggl( \hat{L_o}^{(1)} + S_1 \biggr) + y_p \biggl( S_1 (1 - \alpha_1) P_1^{(1)} + \sum_{m=2}^M S_m p_m^{(1)} \biggr)\\
%\tmu^{(1)} &= y_p \biggl( \hat{L_o}^{(1)} + \sum_{m=1}^{M} S_m p_m^{(1)} \biggr)\\
\Delta \mu^{(1)} &= (y_p - y_o) \hat{L_o}^{(1)} + S_1 (y_p \pi_1-y_o) = y_p (1-\pi_1) \hat{L_o}^{(1)} > 0\\
\end{split}
\end{equation}
while $\Delta \mu^{(j)} = y_p (1-\pi_1) \hat{L_o}^{(j)} > 0 \hspace{1mm} \forall j \in \{2,\cdots,N\}$. The carrier profit gain has can fall in two cases depending on the load:\\
%\begin{equation}
%\begin{split}
%R^* &= y_o \biggl( \sum_{j=1}^N \hat{L_o}^{(j)} + S_1 \biggr) + y_p \biggl( \sum_{j=1}^N \sum_{m=1}^M S_m p_m^{(j)} - S_1 \pi_1 \biggr) \\
        %& - \beta \max{\biggl(\sum_{j=1}^N \hat{L_o}^{(j)} + S_1, \sum_{j=1}^N \sum_{m=1}^M S_m p_m^{(j)} - S_1 \pi_1\biggr)}\\
%\tR &= y_p \sum_{j=1}^N \biggl( \hat{L_o}^{(j)} + \sum_{m=1}^{M} S_m p_m^{(j)} \biggr) - \beta \sum_{j=1}^N \sum_{m=1}^{M} S_m p_m^{(j)}\\
%\end{split}
%\end{equation}
If $\sum_{j=1}^N \sum_{m=1}^{M} S_m p_m^{(j)} < \sum_{j=1}^N \hat{L_o}^{(j)} + S_1 (1+\pi_1)$ Then,
\begin{equation*}
\begin{split}
\Delta R %&= (y_o - y_p - \beta) \sum_{j=1}^N \hat{L_o}^{(j)} + (y_o-y_p \pi_1-\beta) S_1 + \beta \sum_{j=1}^N \sum_{m=1}^M S_m p_m^{(j)} \\
         %&= (y_p \pi_1 - y_p - \beta) \sum_{j=1}^N \hat{L_o}^{(j)} -\beta S_1 + \beta \sum_{j=1}^N \sum_{m=1}^M S_m p_m^{(j)} \\
         &=  y_p(1 - \pi_1) \sum_{j=1}^N \hat{L_o}^{(j)} - \beta \left(\sum_{j=1}^N \hat{L_o}^{(j)} + S_1 \right) + \beta \sum_{j=1}^N \sum_{m=1}^M S_m p_m^{(j)}
\end{split}
\end{equation*}
Hence, a win-win situation happens when:
\begin{equation}\label{Eq_R1_Cond1_NUsers}
\begin{split}
\beta \sum_{j=1}^N \sum_{m=1}^M S_m p_m^{(j)} > \beta \left(\sum_{j=1}^N \hat{L_o}^{(j)} + S_1 \right) + y_p(1 - \pi_1) \sum_{j=1}^N \hat{L_o}^{(j)}
\end{split}
\end{equation}
If $\sum_{j=1}^N \sum_{m=1}^{M} S_m p_m^{(j)} > \sum_{j=1}^N \hat{L_o}^{(j)} + S_1 (1+\pi_1)$ Then,
\begin{equation*}
\begin{split}
\Delta R &= (y_o - y_p) \hat{L_o} + (y_o-y_p \pi_1-\beta \pi_1) S_1 = y_p (\pi_1 - 1) \hat{L_o} -\beta \pi_1 S_1 \\
\end{split}
\end{equation*}
Hence, a win-win situation happens when
\begin{equation}\label{Eq_R1_Cond2_NUsers}
\begin{split}
\beta S_1 \pi_1 > y_p (1 - \pi_1) \sum_{j=1}^N \hat{L_o}^{(j)}
\end{split}
\end{equation}

Notice that if (\ref{Eq_R1_Cond1_NUsers}) is satisfied then we can guarantee that (\ref{Eq_R1_Cond2_NUsers}) is satisfied. Hence, it is sufficient to consider (\ref{Eq_R1_Cond1_NUsers}) for defining the win-win region in this case. i.e.
\begin{equation}
\begin{split}
\mathcal{R}_{1}^{d} &= \Bigg\{ p_m^{(j)}, m \in \{1,2,\cdots,M\}, j \in \{1,2,\cdots,N\} \Big\vert y_o = \hat{y} \pi_1, \\
& \beta \sum_{j=1}^N \sum_{m=1}^M S_m p_m^{(j)} > \beta \left(S_1 +\sum_{j=1}^N \hat{L_o}^{(j)}\right) + \hat{y} (1 - \pi_1) \sum_{j=1}^N \hat{L_o}^{(j)} \Bigg\}\\
\end{split}
\end{equation}

Similarly the carrier can set $y_o=\hat{y} \pi_k$ for any $k \in \{1,2,\cdots,MN-1\}$. This leads to a corresponding $M^{*(j)}$ for each end-user such that $\sum_{j=1}^{N} M^{*(j)} = k$. Hence, we can define the win-win area as follows:
\begin{equation}
\begin{split}
\mathcal{R}_{k}^{d} &= \Bigg\{ p_m^{(j)}, m \in \{1,2,\cdots,M\}, j \in \{1,2,\cdots,N\} \Big\vert y_o = \hat{y} \pi_k,
                    \beta \sum_{j=1}^N \sum_{m=1}^M S_m p_m^{(j)} > \\
                    &\beta \left(\sum_{j=1}^N \sum_{m=1}^{M^{*(j)}} S_m+ \sum_{j=1}^N \hat{L_o}^{(j)} \right)+ \hat{y} (1 - \pi_k) \sum_{j=1}^N \hat{L_o}^{(j)} + \hat{y} \sum_{j=1}^N \sum_{m=1}^{M^{*(j)}} ( \pi_m^{(j)} - \pi_k) S_m \Bigg\}
\end{split}
\end{equation}
Finally, one can define the overall win-win region $\mathcal{R}_{ww}$ to be the region were any win-win situation can happen as follows:
\begin{equation}
\mathcal{R}_{ww}^{d} = \biggl\{ p_m^{(j)}, m \in \{1,2,\cdots,M\}, j \in \{1,2,\cdots,N\} \Big\vert \Delta{\mu^{(j)}}> 0 \hspace{1mm} \forall j, \Delta{R}> 0 \biggr\}
\end{equation}
and the union of all the win-win regions we have found for each carrier strategy has to be a subset of the overall win-win region:
\begin{equation}\label{Eq_WinWin_1_NUsers}
\bigcup_{i=0}^{MN}{\mathcal{R}_{i}^{d}} \subseteq \mathcal{R}_{ww}^{d}
\end{equation}

Now, define $\overline{\mathcal{R}_{ww}^{d}}$ to be the region were no win-win situation can happen and is defined as:
\begin{equation}
\overline{\mathcal{R}_{ww}^{d}} = \biggl\{ p_m^{(j)}, m \in \{1,2,\cdots,M\}, j \in \{1,2,\cdots,N\} \Big\vert \Delta{\mu^{(j)}}\leq 0 \hspace{1mm} \forall j, \Delta{R}\leq 0\biggr\}
\end{equation}

Also $ \forall k \in \{0,MN\}$,
\begin{equation}
\begin{split}
\overline{\mathcal{R}_{k}^{d}} &= \biggl\{ p_m^{(j)}, m \in \{1,2,\cdots,M\}, j \in \{1,2,\cdots,N\} \Big\vert y_o = y_p \pi_{k}, \Delta{\mu^{(j)}}\leq 0 \hspace{1mm} \forall j, \Delta{R}\leq 0\biggr\}\\
\end{split}
\end{equation}

Now, the intersection of all regions of no win-win for each carrier strategy has to be a subset of the over all region of no win-win:
\begin{equation}\label{Eq_WinWin_2_NUsers}
\bigcap_{i=0}^{MN}{\overline{\mathcal{R}_{i}^{d}}} \subseteq  \overline{\mathcal{R}_{ww}^{d}}
\Longleftrightarrow \overline{\biggl(\bigcup_{i=0}^{MN}{\mathcal{R}_{i}^{d}}\biggr)} \subseteq  \overline{\biggl(\mathcal{R}_{ww}^{d}\biggr)}
\Longleftrightarrow \mathcal{R}_{ww}^{d} \subseteq  \bigcup_{i=0}^{MN}{\mathcal{R}_{i}^{d}}
\end{equation}

From (\ref{Eq_WinWin_1_NUsers}) and (\ref{Eq_WinWin_2_NUsers}) we conclude that:
\begin{equation}
\mathcal{R}_{ww}^{d}  = \bigcup_{i=0}^{MN}{\mathcal{R}_{i}^{d}}
\end{equation}

Moreover, we notice that $\mathcal{R}_{1}^{d} \supseteq \mathcal{R}_{2}^{d} \supseteq \mathcal{R}_{3}^{d} \supseteq \cdots. \supseteq \mathcal{R}_{MN-1}^{d}$ , then $\bigcup_{i=0}^{MN}{\mathcal{R}_{i}^{d}} = \mathcal{R}_{1}^{d}$, then
\begin{equation}
\begin{split}
\mathcal{R}_{ww}^{d} &= \Bigg\{ p_m^{(j)}, m \in \{1,2,\cdots,M\}, j \in \{1,2,\cdots,N\} \Big\vert \\
                      & \beta \sum_{j=1}^N \sum_{m=1}^M S_m p_m^{(j)} > \beta \left(S_1 +\sum_{j=1}^N \hat{L_o}^{(j)}\right) + \hat{y} (1 - \pi_1^{(1)}) \sum_{j=1}^N \hat{L_o}^{(j)} \Bigg\}\\
\end{split}
\end{equation}
\end{proof}

\section{Proof of Theorem \ref{Thrm_Connected}}
\begin{proof}
We have $N$ connected end-users consuming data from $M$ data contents. The carrier sets the off-peak price such that only the highest interested user will buy in the peak time as shown in Claim \ref{claim_1}. Suppose that data contents are sorted descendingly according to the interest of the highest interested end-user. Let $\pi_k = \max(\pi_k^{(1)},\pi_k^{(2)},\cdots,\pi_k^{(N)})$ for $k \in \{1,2,\cdots,M\}$. Hence, we get a new set of freshness-popularity products $\Pi=\{\pi_k,k\in\{0,1,2,\cdots,M\}\Big\vert\pi_0>\pi_1>\pi_2>\cdots>\pi_{M}\}$. Therefore, the carrier can simply set $y_0 = \hat{y} \pi_k$ for $k \in \{1,2,\cdots,M\}$. Moreover, the carrier sets $\gamma \leq 1$. We now investigate the win-win region $\mathcal{R}_{k}^{c}$, where $c$ denotes connected, for each value of $y_o$ where $\mathcal{R}_{k}^{d}$ is defined as:
\begin{equation}
\mathcal{R}_{k}^{c} = \biggl\{ p_m^{(j)}, m \in \{1,\cdots,M\}, j \in \{1,\cdots,N\} \Big\vert y_o = \hat{y} \pi_k, \Delta{\mu^{(j)}}> 0 \hspace{1mm} , \Delta{R}> 0 \biggr\}
\end{equation}

\subsection{\textbf{\underline{For $0\leq k \leq M$ and $\gamma=1$:} }}
The carrier is applying smart pricing but when $\gamma=1$ we can see from  (\ref{Eq_OptJStrategy_1}) and (\ref{Eq_xi_opt_1}) that each end-user will cache the data contents if his interest $\pi_{m}^{(j)}>\frac{y_o}{y_p}$ and no one will wait for the highest interested end-user regardless of his selling price. Hence, the win-win region here will be similar to that of the disconnected end-users under the same choice of $y_o$ and $y_p$.
\begin{equation}
\begin{split}
\mathcal{R}^c_{k} =  \mathcal{{R}}_{k}^{d} \ \ \ \ \ 0\leq k \leq M
\end{split}
\end{equation}

\subsection{\textbf{\underline{For $k=0$ and $ \gamma < 1$:}}}
We can see from (\ref{Eq_OptJStrategy_1}) and (\ref{Eq_xi_opt_1}) that no end-user is going to download any data content proactively which means that no trading is going to happen between them. Hence, we have the same scenario as in (\ref{Eq_YoYp_Users}) and (\ref{Eq_YoYp_Carrier}). Thus, $\Delta \mu^{(j)} = 0$ for $j \in \{1,2,\cdots,N\}$ and $\Delta R = 0$. Hence, the win-win region will be the same as the disconnected case. i.e. we have:
\begin{equation}
\begin{split}
\mathcal{R}^c_{0} &= \biggl\{ p_m^{(j)}, m \in \{1,\cdots,M\}, j \in \{1,\cdots,N\} \Big\vert y_o = y_p, \Delta{\mu^{(j)}}> 0 , \Delta{R}> 0\biggr\} = \phi = \mathcal{{R}}_{0}^{d}
\end{split}
\end{equation}

\subsection{\textbf{\underline{For $k>0$ and $\gamma<1$:} }}
In this case, for each source the end-user with the highest interest caches it in the off-peak time while other end-users wait and buy the fresh portion from him in the peak time. Let us assume, without loss of generality, that end-user 1 has the highest interest in data contents $m\in \{1,2,\cdots,k_1\}$ and end-user 2 has the highest interest in data contents $m\in \{k_1+1,k_1+2,\cdots,k_2\}$ and so on. Generally, end-user $j$ has the highest interest in data contents $m\in \{k_{j-1}+1,k_{j-1}+2,\cdots,k_{j}\}$ for $j \in {1,2,\cdots,N}$ where $k_{0}=0$, such that $\sum_{j=1}^{N}(k_j-k_{j-1}) = M$. Hence, the freshness-popularity product set will be,
\begin{equation}
\begin{split}
\Pi &= \left\{\pi_0,\pi_1,\cdots,\pi_{k_1},\pi_{k_1+1},\cdots,\pi_{k_2},\cdots,\pi_{k_{N-1}+1},\cdots,\pi_{k_N}\right\}\\
    &=\left\{\pi_0^{(1)},\pi_1^{(1)},\cdots,\pi_{k_1}^{(1)},\pi_{k_1+1}^{(2)},\cdots,\pi_{k_2}^{(2)},\cdots,\pi_{k_{N-1}+1}^{(N)},\cdots,\pi_{k_N}^{(N)}\right\}
\end{split}
\end{equation}

If the carrier sets $y_o = \hat{y}\pi_{1}^{(1)}$, $y_p = \hat{y}$ and $\gamma<1$, then end-user 1 who is the most interested end-user in this data content reacts by caching it proactively while other end-users wait and buy it from him. Now, lets investigate the possibility of a win-win situation for this case by checking $\Delta \mu^{(j)}$ and $\Delta R$. For user 1 we have:\\
\begin{equation}
\begin{split}
%\mu^{*(1)} &= y_o(\hat{L_o}^{(1)}+S_1) + y_p S_1 (1-\alpha_1) p_1^{(1)} + y_p\sum_{m=2}^{M} S_m p_m^{(1)}-y_s^{(1)}(1-\gamma)\sum_{j=2}^{N} S_1 \pi_1^{(j)}\\
%\tmu^{(1)} &= y_p \left( \hat{L_o}^{(1)} + \sum_{m=1}^{M} S_m p_m^{(1)}\right)\\
\Delta \mu^{(1)} %&= (y_p - y_o) \hat{L_o}^{(1)}  + S_1 \left( y_p \pi_1^{(1)} - y_o + y_s^{(1)}(1-\gamma)\sum_{j=2}^{N} \pi_1^{(j)} \right) \\
&= y_p (1 - \pi_1^{(1)}) \hat{L_o}^{(1)}  + y_s^{(1)} S_1 (1-\gamma)\sum_{j=2}^{N} \pi_1^{(j)}  >0\\
\end{split}
\end{equation}
and $\Delta \mu^{(j)}>0$ for all $j \in \{2,\cdots,N\}$
\begin{equation}
\begin{split}
%\mu^{*(j)} &= y_o \hat{L_o}^{(j)} + y_s^{(1)} S_1 \pi_1^{(j)} + y_p S_1 (1-\alpha_1) p_1^{(j)} + y_p\sum_{m=2}^{M} S_m p_m^{(j)}\\
%\tmu^{(j)} &= y_p \hat{L_o}^{(j)} + y_p S_1 \pi_1^{(j)} + y_p S_1 (1-\alpha_1) p_1^{(j)} + y_p\sum_{m=2}^{M} S_m p_m^{(j)}\\
\Delta \mu^{(j)} %&= (y_p - y_o) \hat{L_o}^{(j)}  + (y_p - y_s^{(1)}) S_1  \pi_1^{(j)} \\
&= y_p(1 - \pi_1^{(1)}) \hat{L_o}^{(j)}  + (y_p - y_s^{(1)}) S_1  \pi_1^{(j)} > 0\\
\end{split}
\end{equation}
The carrier profit gain is:\\
%\begin{equation}
%\begin{split}
%R^* &= y_o\sum_{\substack{j=1 }}^{N} \hat{L_o}^{(j)} + y_o S_1 + y_p S_1\sum_{j=1}^{N} (1-\alpha_1) p_1^{(j)} + y_p\sum_{j=1}^{N}\sum_{m=2}^{M} S_m p_m^{(j)} +y_s^{(1)}\gamma\sum_{j=2}^{N} S_1 \pi_1^{(j)}\\
% - \beta \max{ \left( \sum_{\substack{j=1}}^{N} \hat{L_o}^{(j)} + S_1, S_1 \sum_{j=1}^{N} (1-\alpha_1) p_1^{(j)} + \sum_{j=1}^{N}\sum_{m=2}^{M} S_m p_m^{(j)} \right) }\\
%\tR &= y_p\sum_{\substack{j=1 }}^{N} \hat{L_o}^{(j)} +  y_p S_1\sum_{j=1}^{N} p_1^{(j)} + y_p\sum_{j=1}^{N}\sum_{m=2}^{M} S_m p_m^{(j)} - \beta \sum_{j=1}^{N}\sum_{m=1}^{M} S_m p_m^{(j)}\\
%\end{split}
%\end{equation}
%if $\sum_{\substack{j=1 }}^{N} \hat{L_o}^{(j)} + S_1 +S_1 \sum_{j=1}^{N} \pi_1^{(j)}  > \sum_{j=1}^{N}\sum_{m=1}^{M} S_m p_m^{(j)}$, then
\begin{equation}
\begin{split}
\Delta R %&= (y_o - y_p)\sum_{\substack{j=1 }}^{N} \hat{L_o}^{(j)} + S_1 \left( y_o - y_p \pi_1^{(1)} + (\gamma y_s^{(1)}  - y_p) \sum_{j=2}^{N}  \pi_1^{(j)} \right)
%- \beta \left( \sum_{\substack{j=1 }}^{N} \hat{L_o}^{(j)} + S_1 - \sum_{j=1}^{N}\sum_{m=1}^{M} S_m p_m^{(j)} \right)\\
&= y_p (\pi_1^{(1)} -1)\sum_{\substack{j=1 }}^{N} \hat{L_o}^{(j)} + y_p S_1 \left( \frac{\pi_{1}^{(1)}}{\max\limits_{\substack{j\in \{2,\cdots,N\}}}\pi_{1}^{(j)}}-1 \right) \sum_{j=2}^{N} \pi_1^{(j)} \\
&- \beta \left( \sum_{\substack{j=1 }}^{N} \hat{L_o}^{(j)} + S_1 - \sum_{j=1}^{N}\sum_{m=1}^{M} S_m p_m^{(j)} \right)\\
\end{split}
\end{equation}

Hence, a win-win situation is:
\begin{equation}
\begin{split}
\mathcal{R}_{1}^{c} &= \Biggl\{ p_m^{(j)}, m \in \{1,2,\cdots,M\}, j \in \{1,2,\cdots,N\} \Big\vert y_o = \hat{y} \pi_1^{(1)}, \beta \sum_{j=1}^{N} \sum_{m=1}^{M} S_m p_m^{(j)} > \\
& \beta \left(S_1 + \sum_{\substack{j=1 }}^{N} \hat{L_o}^{(j)} \right) + \hat{y}(1 - \pi_{1}^{(1)})\sum_{\substack{j=1 }}^{N} \hat{L_o}^{(j)} + \hat{y} S_1 \left( 1-\frac{\pi_{1}^{(1)}}{\max\limits_{\substack{j\in \{2,\cdots,N\}}}\pi_{1}^{(j)}} \right) \sum_{j=2}^{N} \pi_1^{(j)} \Biggr\}\\
\end{split}
\end{equation}

%Similarly if the carrier sets $y_o = \hat{y} \pi_{2}^{(1)}$, $y_p = \hat{y}$ and $\gamma<1$, then end-user 1 who is still the most interested end-user in these data contents will react by caching it proactively while other end-users will wait and buy it from him. By the same token we have the win-win region as:
%\begin{equation}
%\begin{split}
%\mathcal{R}_{2}^{c} &= \Bigg\{ p_m^{(j)}, m \in \{1,2,\cdots,M]\}, j \in \{1,2,\cdots,N\} \Big\vert y_o = \hat{y} \pi_2^{(1)}, \\
 %                   & \beta \sum_{m=1}^{M} \sum_{j=1}^{N} S_m p_m^{(j)} > \beta \left(S_1 + S_2 + \sum_{\substack{j=1 }}^{N} \hat{L_o}^{(j)} \right) +\hat{y}(1 - \pi_{2}^{(1)})\sum_{\substack{j=1}}^{N} \hat{L_o}^{(j)} \\
%%                    & + \hat{y} (\pi_1^{(1)} - \pi_2^{(1)}) S_1
%                    + \hat{y} \left( 1-\frac{\pi_{2}^{(1)}}{\max\limits_{\substack{j\in \{2,\cdots,N\} \\ m\in \{1,2\}}}\pi_{m}^{(j)}} \right) \sum_{j=2}^{N} \sum_{m=1}^{2} S_m \pi_m^{(j)} \Bigg\}
%\end{split}
%\end{equation}

If the carrier sets $y_o = \hat{y} \pi_{k_1}^{(1)}$, $y_p = \hat{y}$ and $\gamma<1$, then end-user 1 caches data contents $m \in \{1,2,\cdots,k_1\}$ proactively while other end-users wait and buy it from him. Hence, the win-win region is:
\begin{equation}
\begin{split}
\mathcal{R}_{k_1}^{c} &= \Bigg\{ p_m^{(j)}, m \in \{1,2,\cdots,M\}, j \in \{1,2,\cdots,N\} \Big\vert y_o = \hat{y} \pi_{k_1}^{(1)}, \\
                    & \beta \sum_{j=1}^{N} \sum_{m=1}^{M} S_m p_m^{(j)} > \beta \left(\sum_{m=1}^{k_1}S_m + \sum_{\substack{j=1 }}^{N} \hat{L_o}^{(j)} \right) + \hat{y} (1 - \pi_{k_1}^{(1)})\sum_{\substack{j=1 }}^{N} \hat{L_o}^{(j)} \\
& + \hat{y} \sum_{m=1}^{k_1} (\pi_m^{(1)} - \pi_{k_1}^{(1)}) S_m + \hat{y} \left(1-\frac{\pi_{k_1}^{(1)}}{\max\limits_{\substack{j\in \{2,\cdots,N\} \\ m\in \{1,\cdots,k_1\}}}\pi_{m}^{(j)}} \right) \sum_{\substack{j=2 }}^{N} \sum_{m=1}^{k_1}S_m  \pi_m^{(j)} \Bigg\}
\end{split}
\end{equation}

If the carrier sets $y_o = \hat{y} \pi_{k_2}^{(2)}$ and $y_p = \hat{y}$, then end-user 1 reacts by caching data contents $m \in \{1,2,\cdots,k_1\}$ and sells it to all other end-users. Moreover, end-user 2 reacts by caching data contents $m \in \{k_1+1,k_1+2,\cdots,k_2\}$ and sells it to all other end-users. Hence for end-user 1 we have:
\begin{equation}
\begin{split}
%\mu^{*(1)} &= y_o\left(\hat{L_o}^{(1)}+\sum_{m=1}^{k_1}S_m\right) + y_s^{(2)}\sum_{m=k_1+1}^{k_2} S_m \pi_m^{(1)} + y_p\sum_{m=k_2+1}^{M} S_m p_m^{(1)}\\
%& + y_p \sum_{m=1}^{k_2}S_m (1-\alpha_m) p_m^{(1)}-y_s^{(1)}(1-\gamma)\sum_{j=2}^{N} \sum_{m=1}^{k_1}S_m \pi_m^{(j)}\\
%\tmu^{(1)} &= y_p \left( \hat{L_o}^{(1)} + \sum_{m=1}^{M}S_m p_m^{(1)} \right) = y_p\hat{L_o}^{(1)} + y_p\sum_{m=1}^{k_2}S_m p_m^{(1)} + y_p\sum_{m=k_2+1}^{M} S_m p_m^{(1)}\\
\Delta \mu^{(1)} %&= (y_p - y_o) \hat{L_o}^{(1)} + \sum_{m=1}^{k_1} S_m (y_p \pi_m^{(1)} - y_o) + (y_p - y_s^{(2)}) \sum_{m=k_1+1}^{k_2}S_m \pi_m^{(1)} + y_s^{(1)}(1-\gamma)\sum_{j=2}^{N}\sum_{m=1}^{k_1} S_m \pi_m^{(j)}\\
&= y_p (1 - \pi_{k_2}^{(2)}) \hat{L_o}^{(1)} + y_p \sum_{m=1}^{k_1} S_m (\pi_m^{(1)}-\pi_{k_2}^{(2)}) + (y_p - y_s^{(2)}) \sum_{m=k_1+1}^{k_2}S_m \pi_m^{(1)}\\
&+ y_s^{(1)}(1-\gamma)\sum_{j=2}^{N}\sum_{m=1}^{k_1} S_m \pi_m^{(j)} > 0\\
%&= y_p \hat{L_o}^{(1)} (1 - \pi_{k_2}^{(2)}) +y_p \frac{\pi_{k_1}^{(1)}}{\max\limits_{\substack{j\in \{2,\cdots,N\} \\ m\in \{1,\cdots,k_1+k_2\}}}\pi_{m}^{(j)}}(1-\gamma)\sum_{m=1}^{k_1}\sum_{i=2}^{N} S_m \pi_m^{(j)}+ y_p \sum_{m=1}^{k_2} S_m \pi_m^{(1)} - y_p \pi_{k_2}^{(2)} \sum_{m=1}^{k_1} S_m
%\\&- y_p \frac{\pi_{k_2}^{(2)}}{\max\limits_{\substack{j\in \{1,3,\cdots,N\} \\ m\in \{k_1+1,\cdots,k_1+k_2\}}}\pi_{m}^{(j)}} \sum_{m=k_1+1}^{k_2}S_m \pi_m^{(1)}
\end{split}
\end{equation}

For end-user 2 we have:
\begin{equation}
\begin{split}
%\mu^{*(2)} &= y_o\left(\hat{L_o}^{(2)}+\sum_{m=k_1+1}^{k_2}S_m\right) + y_s^{(1)}\sum_{m=1}^{k_1} S_m \pi_m^{(2)} + y_p\sum_{m=k_2+1}^{M} S_m p_m^{(2)}\\
%& + y_p \sum_{m=1}^{k_2}S_m (1-\alpha_m) p_m^{(2)} -y_s^{(2)}(1-\gamma)\sum_{\substack{j=1 \\ j\neq2}}^{N} \sum_{m=k_1+1}^{k_2}S_m \pi_m^{(j)}\\
%\tmu^{(2)} &= y_p \left( \hat{L_o}^{(2)} + \sum_{m=1}^{M}S_m p_m^{(2)} \right) = y_p\hat{L_o}^{(2)} + y_p\sum_{m=1}^{k_2}S_m p_m^{(2)} + y_p\sum_{m=k_2+1}^{M} S_m p_m^{(2)}\\
\Delta \mu^{(2)} %&= (y_p - y_o)\hat{L_o}^{(2)} + (y_p - y_s^{(1)})\sum_{m=1}^{k_1}S_m \pi_m^{(2)} + \sum_{m=k_1+1}^{k_2}S_m (y_p\pi_m^{(2)} -y_o)+ y_s^{(2)}(1-\gamma)\sum_{\substack{j=1 \\ j\neq2}}^{N} \sum_{m=k_1+1}^{k_2}S_m \pi_m^{(j)}\\
&= y_p (1 - \pi_{k_2}^{(2)})\hat{L_o}^{(2)} + (y_p - y_s^{(1)})\sum_{m=1}^{k_1}S_m \pi_m^{(2)} +  y_p\sum_{m=k_1+1}^{k_2}(\pi_m^{(2)}-\pi_{k_2}^{(2)}) S_m\\
&+ y_s^{(2)}(1-\gamma)\sum_{\substack{j=1 \\ j\neq2}}^{N} \sum_{m=k_1+1}^{k_2}S_m \pi_m^{(j)} > 0
\end{split}
\end{equation}

Now, since all other end-users $j\in\{3,4,\cdots,N\}$ are symmetric, it is sufficient to check the saving of one of them and show that $\Delta \mu^{(j)}>0$ for all $j \in \{3,4,\cdots,N\}$
\begin{equation}
\begin{split}
%\mu^{*(j)} &= y_o \hat{L_o}^{(j)} + y_s^{(1)}\sum_{m=1}^{k_1} S_m \pi_m^{(j)} + y_s^{(2)}\sum_{m=k_1+1}^{k_2} S_m \pi_m^{(j)}+ y_p\sum_{m=k_2+1}^{M} S_m p_m^{(j)} + y_p \sum_{m=1}^{k_2} S_m(1-\alpha_m) p_m^{(j)}\\
%\tmu^{(j)} &= y_p \left( \hat{L_o}^{(j)} + \sum_{m=1}^{M}S_m p_m^{(j)} \right) = y_p\hat{L_o}^{(j)} + y_p\sum_{m=1}^{k_2}S_m p_m^{(j)} + y_p\sum_{m=k_2+1}^{M} S_m p_m^{(j)}\\
\Delta \mu^{(j)} %&= (y_p - y_o)\hat{L_o}^{(j)} + (y_p - y_s^{(1)})\sum_{m=1}^{k_1}S_m \pi_m^{(2)} + (y_p - y_s^{(2)})\sum_{m=k_1+1}^{k_2}S_m \pi_m^{(j)}\\
&= y_p(1 - \pi_{k_2}^{(2)}) \hat{L_o}^{(j)} + (y_p - y_s^{(1)}) \sum_{m=1}^{k_1} S_m \pi_m^{(j)}+ (y_p- y_s^{(2)}) \sum_{m=k_1+1}^{k_2} S_m \pi_m^{(j)} >0 \\
\end{split}
\end{equation}

The carrier profit gain is:\\
%\begin{equation}
%\begin{split}
%R^{*} &= y_o\sum_{\substack{j=1 }}^{N} \hat{L_o}^{(j)} + y_o \sum_{m=1}^{k_2}S_m + y_p \sum_{j=1}^{N} \sum_{m=1}^{k_2}  S_m(1-\alpha_m) p_m^{(j)} + y_p\sum_{j=1}^{N}\sum_{m=k_2+1}^{M} S_m p_m^{(j)} + \gamma y_s^{(1)}\sum_{j=2}^{N} \sum_{m=1}^{k_1} S_m \pi_m^{(j)}\\
%& + \gamma y_s^{(2)}\sum_{\substack{j=1 \\ j\neq2}}^{N} \sum_{m=k_1+1}^{k_2} S_m \pi_m^{(j)}- \beta \max{\left(\sum_{\substack{j=1}}^{N} \hat{L_o}^{(j)} + \sum_{m=1}^{k_2} S_m, \sum_{j=1}^{N}\sum_{m=1}^{k_2} S_m (1-\alpha_m) p_m^{(j)} + \sum_{j=1}^{N}\sum_{m=k_2+1}^{M} S_m p_m^{(j)}\right)}\\
%\tR &= y_p\sum_{\substack{j=1 }}^{N} \hat{L_o}^{(j)} +  y_p \sum_{j=1}^{N}\sum_{m=1}^{k_2}S_m p_m^{(j)} + y_p\sum_{j=1}^{N}\sum_{m=k_2+1}^{M} S_m p_m^{(j)} - \beta \sum_{j=1}^{N}\sum_{m=1}^{M} S_m p_m^{(j)}\\
%\Delta R &= (y_o - y_p)\sum_{\substack{j=1 }}^{N} \hat{L_o}^{(j)} - y_p \sum_{j=1}^{N} \sum_{m=1}^{k_2} S_m \pi_m^{(j)} + \gamma y_s^{(1)} \sum_{j=2}^{N} \sum_{m=1}^{k_1}  S_m \pi_m^{(j)} + \gamma y_s^{(2)} \sum_{\substack{j=1 \\ j\neq2}}^{N} \sum_{m=k_1+1}^{k_2}  S_m \pi_m^{(j)}\\
%& + y_o \sum_{m=1}^{k_2}S_m + \beta \sum_{j=1}^{N}\sum_{m=1}^{M} S_m p_m^{(j)} - \beta \max{\left(\sum_{\substack{j=1 }}^{N} \hat{L_o}^{(j)} + \sum_{m=1}^{k_2}S_m, \sum_{j=1}^{N}\sum_{m=1}^{M} S_m p_m^{(j)} - \sum_{j=1}^{N} \sum_{m=1}^{k_2} S_m \pi_m^{(j)}\right)}
%\end{split}
%\end{equation}

%if $\sum_{\substack{j=1 }}^{N} \hat{L_o}^{(j)} + \sum_{m=1}^{k_2}S_m - \sum_{j=1}^{N}\sum_{m=1}^{k_2} S_m \pi_m^{(j)} > \sum_{j=1}^{N}\sum_{m=1}^{M} S_m p_m^{(j)}$, then:
\begin{equation}
\begin{split}
\Delta R &= y_p\left(\pi_{k_2}^{(2)} - 1 \right)\sum_{\substack{j=1 }}^{N} \hat{L_o}^{(j)} + y_p \pi_{k_2}^{(2)} \sum_{m=1}^{k_2}S_m - y_p \sum_{j=1}^{N}\sum_{m=1}^{k_2} S_m \pi_m^{(j)} + y_p\left(\frac{\pi_{k_1}^{(1)}}{\max\limits_{\substack{j\in \{2,\cdots,N\} \\ m\in \{1,\cdots,k_1\}}}\pi_{m}^{(j)}}\right) \sum_{j=2}^{N}\sum_{m=1}^{k_1} S_m \pi_m^{(j)} \\
&+  y_p \left(\frac{\pi_{k_2}^{(2)}}{\max\limits_{\substack{j\in \{1,3,\cdots,N\} \\ m\in \{k_1+1,\cdots,k_1+k_2\}}}\pi_{m}^{(j)}}\right) \sum_{\substack{j=1 \\ j\neq2}}^{N} \sum_{m=k_1+1}^{k_2} S_m \pi_m^{(j)} + \beta \sum_{j=1}^{N}\sum_{m=1}^{M} S_m p_m^{(j)} - \beta \left(\sum_{\substack{j=1 }}^{N} \hat{L_o}^{(j)} + \sum_{m=1}^{k_2}S_m \right)\\
\end{split}
\end{equation}

Hence, a win-win situation is:
%\begin{equation}\label{Eq_RK1_Cond2_NUsers_trading}
%\begin{split}
%&\beta \sum_{m=1}^{k_2} \sum_{i=1}^{N} S_m \pi_m^{(j)}  > y_p\left(1 - \pi_{k_2}^{(2)}\right)\sum_{\substack{j=1 }}^{N} \hat{L_o}^{(j)} + y_p \sum_{m=1}^{k_1} (\pi_m^{(1)} - \pi_{k_2}^{(2)}) S_m + y_p \sum_{m=k_1+1}^{k_2} (\pi_m^{(2)} - \pi_{k_2}^{(2)}) S_m\\
%& + y_p\left(1-\frac{\pi_{k_1}^{(1)}}{\max\limits_{\substack{j\in \{2,\cdots,N\} \\ m\in \{1,\cdots,k_1\}}}\pi_{m}^{(j)}}\right) \sum_{j=2}^{N}\sum_{m=1}^{k_1} S_m \pi_m^{(j)} + y_p \left(1-\frac{\pi_{k_2}^{(2)}}{\max\limits_{\substack{j\in \{1,3,\cdots,N\} \\ m\in \{k_1+1,\cdots,k_1+k_2\}}}\pi_{m}^{(j)}}\right) \sum_{\substack{j=1 \\ j\neq2}}^{N} \sum_{m=k_1+1}^{k_2} S_m \pi_m^{(j)}
%\end{split}
%\end{equation}
%Notice that if (\ref{Eq_RK1_Cond1_NUsers_trading}) was satisfied then we can guarantee that (\ref{Eq_RK1_Cond2_NUsers_trading}) will be satisfied. Hence, it is sufficient to consider (\ref{Eq_RK1_Cond1_NUsers_trading}) for defining the win-win area in this case. i.e.
\begin{equation}\label{Eq_RK2_NUsers_trading}
\begin{split}
\mathcal{R}^c_{k_2} &= \Biggl\{ p_m^{(j)}, m \in \{1,2,\cdots,M\}, j \in \{1,2,\cdots,N\} \Big\vert \\
& \beta \sum_{j=1}^{N} \sum_{m=1}^{M} S_m p_m^{(j)} > \beta \left(\sum_{\substack{j=1 }}^{N} \hat{L_o}^{(j)} + \sum_{m=1}^{k_2}S_m \right) + \hat{y}(1 - \pi_{k_2}^{(2)})\sum_{\substack{j=1 }}^{N} \hat{L_o}^{(j)} \\
& + \hat{y} \sum_{m=1}^{k_1} (\pi_m^{(1)} - \pi_{k_2}^{(2)}) S_m + \hat{y} \sum_{m=k_1+1}^{k_2} (\pi_m^{(2)} - \pi_{k_2}^{(2)}) S_m\\
& + \hat{y}\left(1-\frac{\pi_{k_1}^{(1)}}{\max\limits_{\substack{j\in \{2,\cdots,N\} \\ m\in \{1,\cdots,k_1\}}}\pi_{m}^{(j)}}\right) \sum_{j=2}^{N}\sum_{m=1}^{k_1} S_m \pi_m^{(j)} + \hat{y} \left(1-\frac{\pi_{k_2}^{(2)}}{\max\limits_{\substack{j\in \{1,3,\cdots,N\} \\ m\in \{k_1+1,\cdots,k_1+k_2\}}}\pi_{m}^{(j)}}\right) \sum_{\substack{j=1 \\ j\neq2}}^{N} \sum_{m=k_1+1}^{k_2} S_m \pi_m^{(j)}\Biggr\}\\
\end{split}
\end{equation}

The win-win region defined in (\ref{Eq_RK2_NUsers_trading}) can be expanded to the case when the carrier sets $y_o = \hat{y} \pi_{M}^{(N)}$, $y_p = \hat{y}$ and $\gamma <1$. Each end-user $j\in\{1,\cdots,N\}$ reacts by downloading data contents $m \in \{k_{j-1}+1,k_{j-1}+2,\cdots,k_{j}\}$ and sells it to all other users. Hence, the win-win region will be:
\begin{equation}\label{Eq_RM_NUsers_trading}
\begin{split}
\mathcal{R}^c_{M} &= \Biggl\{ p_m^{(j)}, m \in \{1,2,\cdots,M\}, j \in \{1,2,\cdots,N\} \Big\vert \\
& \beta \sum_{j=1}^{N} \sum_{m=1}^{M} S_m p_m^{(j)} > \beta \left(\sum_{\substack{j=1 }}^{N} \hat{L_o}^{(j)} + \sum_{m=1}^{k_2}S_m \right) + \hat{y}(1 - \pi_{k_2}^{(2)})\sum_{\substack{j=1 }}^{N} \hat{L_o}^{(j)} \\
& + \hat{y} \sum_{j=1}^{N} \sum_{m=k_{j-1}+1}^{k_j} (\pi_m^{(j)} - \pi_{M}^{(N)}) S_m + \hat{y} \sum_{j=1}^{N} \left(1-\frac{\pi_{k_j}^{(j)}}{\max\limits_{\substack{i\in \{1,\cdots,N\}, i \neq j \\ m\in \{k_{j-1}+1,\cdots,k_j\}}}\pi_{m}^{(i)}}\right) \sum_{\substack{i=1 \\ i\neq j}}^{N}\sum_{m=1}^{k_1} S_m \pi_m^{(i)} \Biggr\}\\
\end{split}
\end{equation}

Now, let us define $\mathcal{R}_{ww}^{c}$ to be the region were a win-win situation can happen and is defined as:
\begin{equation}
\mathcal{R}_{ww}^{c} = \left\{p_m^{(j)}, m \in \{1,2,\cdots,M\}, j \in \{1,2,\cdots,N\} \Big\vert \Delta{\mu}^{(j)}> 0, \Delta{R}> 0\right\}
\end{equation}

Hence, the union of all the win-win regions we have found for each carrier strategy has to be a subset of the overall win-win region:
\begin{equation}\label{Eq_winwin_Connected_1}
\bigcup_{m=1}^{M}{\mathcal{R}_{m}^{c}} \subseteq \mathcal{R}_{ww}^{c}
\end{equation}

Now, define $\overline{\mathcal{R}_{ww}^{c}}$ to be the region were no win-win situation can happen and is defined as:
\begin{equation}
\overline{\mathcal{R}_{ww}^{c}} = \left\{p_m^{(j)}, m \in \{1,2,\cdots,M\}, j \in \{1,2,\cdots,N\} \Big\vert \Delta{\mu}^{(j)} \leq 0 , \Delta{R}\leq 0\right\}
\end{equation}

Also,
\begin{equation}
\begin{split}
\overline{\mathcal{R}_{m}^{c}} &= \left\{p_m^{(j)}, m \in \{1,2,\cdots,M\}, j \in \{1,2,\cdots,N\} \Big\vert y_o = y_p \pi_{m}^{(j)}, \Delta{\mu}^{(j)} \leq 0 , \Delta{R}\leq 0\right\}\\
%&\vdots\\
%\overline{\mathcal{R}_{k_1}^{c}} &= \left\{p_m^{(j)}, m \in \{1,2,\cdots,M\}, j \in \{1,2,\cdots,N\} \Big\vert y_o = y_p \pi_{k_1}^{(1)}, \Delta{\mu}^{(j)} \leq 0 , \Delta{R}\leq 0\right\}\\
%&\vdots\\
%\overline{\mathcal{R}_{k_2}^{c}} &= \left\{p_m^{(j)}, m \in \{1,2,\cdots,M\}, j \in \{1,2,\cdots,N\} \Big\vert y_o = y_p \pi_{k_2}^{(2)}, \Delta{\mu}^{(j)} \leq 0 , \Delta{R}\leq 0\right\}\\
%&\vdots\\
%\overline{\mathcal{R}_{M}^{c}} &= \left\{p_m^{(j)}, m \in \{1,2,\cdots,M\}, j \in \{1,2,\cdots,N\} \Big\vert y_o = y_p \pi_{M}^{(N)}, \Delta{\mu}^{(j)} \leq 0 , \Delta{R}\leq 0\right\}\\
\end{split}
\end{equation}

Now, the intersection of all regions of no win-win for each carrier strategy has to be a subset of the over all region of no win-win: 
\begin{equation}\label{Eq_winwin_Connected_2}
\bigcap_{m=1}^{M}\overline{\mathcal{R}_{m}^{c}} \subseteq  \overline{\mathcal{R}_{ww}^{c}} \nonumber
\Leftrightarrow \biggl(\overline{\bigcup_{m=1}^{M}{\mathcal{R}^c_{m}}\biggr)}\subseteq   \overline{\biggl(\mathcal{R}^c_{ww}\biggr)} \nonumber
\Leftrightarrow \mathcal{R}^c_{ww}  \subseteq  \bigcup_{m=1}^{M}{\mathcal{R}^c_{m}}
\end{equation}

From (\ref{Eq_winwin_Connected_1}) and (\ref{Eq_winwin_Connected_1}) we conclude that:
\begin{equation}\label{win-win13}
\mathcal{R}^c_{ww}  =  \bigcup_{m=1}^{M}{\mathcal{R}^c_{m}}
\end{equation}
Moreover, (\ref{win-win13}) can be simplified more if we notice that $\mathcal{R}^c_{1} \supseteq \mathcal{R}^c_{2} \supseteq \mathcal{R}^c_{3} \supseteq \cdots. \supseteq \mathcal{R}^c_{M}$ , then $\bigcup_m{\mathcal{R}^c_{m}} = \mathcal{R}^c_{1}$ and the win-win region is given by:
\begin{equation}\label{win-win14}
\begin{split}
\mathcal{R}^c_{win-win}  &=  {\mathcal{R}^c_{1}}\\
 &= \Biggl\{ p_m^{(j)}, m \in \{1,2,\cdots,M\}, j \in \{1,2,\cdots,N\} \Big\vert \beta \sum_{j=1}^{N} \sum_{m=1}^{M} S_m p_m^{(j)} > \\
& \beta \left(S_1 + \sum_{\substack{j=1 }}^{N} \hat{L_o}^{(j)} \right) + \hat{y}(1 - \pi_{1}^{(1)})\sum_{\substack{j=1 }}^{N} \hat{L_o}^{(j)} + \hat{y} S_1 \left( 1-\frac{\pi_{1}^{(1)}}{\max\limits_{\substack{j\in \{2,\cdots,N\}}}\pi_{1}^{(j)}} \right) \sum_{j=2}^{N} \pi_1^{(j)} \Biggr\}\\
\end{split}
\end{equation}
\end{proof}
\end{appendices}

\begin{comment}
\begin{algorithm}
\caption{Carrier's Pricing Strategy Selection}
\label{Strategy_Algorithm}
\begin{algorithmic}[1]
\State Given: $N,M,p_m^{(j)},\alpha_m$
\Function{optYpTrading}{$N,M,p_m^{(j)},\alpha_m$}
    \State $\pi_k = \max(\pi_k^{(1)},\pi_k^{(2)},\cdots,\pi_k^{(N)})$
    \State $\Pi=\{\pi_k,k\in\{0,1,2,\cdots,M\}\Big\vert\pi_0>\pi_1>\pi_2>\cdots>\pi_{M}\}$
    \For{k = 0 to M}
        \State $y_o= \hat{y} \pi_k$, $y_p=\hat{y}$
        \For{$\gamma$ = 0 to 1}
            \State {Calculate $\Delta R$}
        \EndFor
    \EndFor
    \State $(y_p^*,\gamma*)=\argmax \Delta R$
\EndFunction

\Function{optYpNoTrading}{$N,M,p_m^{(j)},\alpha_m$}
\State $\Pi=\{\pi_k,k\in\{0,1,2,\cdots,MN\}\Big\vert\pi_0>\pi_1>\pi_2>\cdots>\pi_{MN}\}$
    \For{k = 0 to MN}
        \State $y_o= \hat{y} \pi_k$, $y_p=\hat{y}$
        \For{$\gamma$ = 0 to 1}
            \State {Calculate $\Delta R$}
        \EndFor
    \EndFor
    \State $(y_p^*,\gamma*)=\argmax \Delta R$
\EndFunction

\Procedure{Win-Win Check}{}
    \If {win-win condition under trading}
        \State $y_o^*=\hat{y}$, $(y_p^*,\gamma^*)$ = \Call{optYpTrading}{$N,M,p_m^{(j)},\alpha_m$}
    \ElsIf {win-win condition without trading}
        \State $y_o^*=\hat{y}$, $(y_p^*,\gamma^*)$ = \Call{optYpNoTrading}{$N,M,p_m^{(j)},\alpha_m$}
    \Else
        \State $y_o^*=y_p^*=\hat{y},\gamma^*=0$
    \EndIf
\EndProcedure

\end{algorithmic}
\end{algorithm}
\end{comment}

\end{document}